\documentclass[twocolumn]{autart}

\usepackage{amsmath,amssymb,amsfonts}
\usepackage{mathtools}
\usepackage{cases}
\usepackage{esdiff}

\usepackage{graphicx}
\usepackage{array}
\usepackage{multirow}
\usepackage{float}      % enables [H]
\usepackage{stfloats}   % improves two-column float handling
\usepackage{placeins}   % \FloatBarrier
\usepackage{xcolor}
\usepackage{textcomp}
\usepackage{caption}
\usepackage{subcaption}

\usepackage{algorithm}
\usepackage{algpseudocode}

\usepackage{cite}
\usepackage{url}

\usepackage[hidelinks,bookmarks=false]{hyperref}

\usepackage{comment}
\usepackage{verbatim}
\usepackage{multicol}
\usepackage{balance}
\usepackage{etoolbox}

\usepackage{tikz}
\usetikzlibrary{arrows.meta,positioning,calc}

\newtheorem{remark}{Remark}
\newtheorem{assumption}{Assumption}
\newtheorem{lemma}{Lemma}

\newtheorem{theorem}{Theorem}

\newenvironment{proof}{\par\noindent\textit{Proof.}\ }{\hfill$\square$\par}

\begin{document}

\begin{frontmatter}

\title{Dynamic Quantum Optimal Communication Topology Design for Consensus Control in Linear Multi-Agent Systems} % Title, preferably not more 
                                                % than 10 words.

\thanks[footnoteinfo]{This paper was not presented at any IFAC 
meeting. This work was supported by the National Science Foundation under
Grant ECCS-1944752 and Grant ECCS-2312086. Corresponding author A.~Kargarian.}

\author[LSU]{Milad Hasanzadeh}\ead{mhasa42@lsu.edu},    % Add the 
\author[LSU]{Amin Kargarian}\ead{kargarian@lsu.edu}     % e-mail address 

\address[LSU]{Department of Electrical Engineering, Louisiana State University, 
Louisiana, USA}  % Please supply full address here.

                             % keyword list or 

\begin{abstract}            
This paper proposes a quantum framework for the design of communication topologies in consensus-based multi-agent systems. The communication graph is selected online by solving a mixed-integer quadratic program (MIQP) that minimizes a cost combining communication and distance penalties with degree-regularization terms, while enforcing exact connectivity through a flow-based formulation. To cope with the combinatorial complexity of this NP-hard problem, we develop a three-block ADMM scheme that decomposes the MIQP into a convex quadratic program in relaxed edge and flow variables, a pure binary unconstrained subproblem, and a closed-form auxiliary update. The binary subproblem is mapped to a quadratic unconstrained binary optimization (QUBO) Hamiltonian and approximately solved via quantum imaginary time evolution (QITE). The resulting time-varying, optimizer-generated Laplacians are applied to linear first- and second-order consensus dynamics. Numerical simulations on networks demonstrate that the proposed method produces connected topologies that satisfy degree constraints, achieve consensus, and incur costs comparable to those of classical mixed-integer solvers, thereby illustrating how quantum algorithms can be embedded as topology optimizers within closed-loop distributed control architectures.
\end{abstract}

\end{frontmatter}

\section{Introduction} \label{sec:intro}
\subsection{Background and Motivation} \label{subsec:related_work}
Multi-agent systems (MASs) consist of a collection of interacting agents that cooperate to achieve collective objectives through local communication and distributed control \cite{kargarian2016toward}. A central problem in MASs is consensus, the process by which agents agree on quantities by exchanging information with neighbors over a communication network \cite{olfati2007consensus}. In early research on consensus, the communication links between agents were fixed and did not change during the task. This setup is called a {static topology}, in which the communication network structure remains the same over time \cite{ hasanzadeh2024dynamic, hasanzadeh2024distributed}. To make the system more flexible, later studies introduced time-varying or switching topologies, called {dynamic topologies}, in which the connections between agents change over time according to schedules \cite{su2015distributed}.

To improve system performance, researchers have studied how to design an optimized communication network for MASs. This is called optimal topology design, where the connections between agents are chosen or updated based on goals such as fast consensus, low communication cost, or strong network reliability \cite{ rafiee2010optimal,yang2010decentralized,aragues2014distributed}. In previous works, the communication graph is chosen by maximizing algebraic connectivity or by enforcing a lower bound on the Fiedler eigenvalue of the Laplacian, which leads to mixed-integer semidefinite or spectral optimization problems \cite{rafiee2010optimal,somisetty2024optimal,tegling2023scale}. These formulations are mathematically elegant and provide direct control over a spectral performance index, but they also have some limitations for practical use. First, they optimize an eigenvalue proxy rather than physical quantities such as communication cost, link distance, or security risk, which makes the trade-offs less transparent. Second, semidefinite constraints and eigenvalue inequalities are computationally heavy, especially if the topology must be recomputed repeatedly in a closed loop or for many different operating points. 

Motivated by these limitations, a different modeling approach that is closer in spirit to classical minimum-cost network design and degree-constrained spanning tree problems can be adopted \cite{goemans2006minimum}. The resulting problem is a mixed-integer quadratic program (MIQP). While effective, such approaches also quickly become computationally expensive as the number of agents and potential links grows due to the nonconvex nature of the problem \cite{burer2009nonconvex,vielma2015mixed}. Exact solvers scale poorly for these NP-hard problems, and heuristic methods trade optimality for tractability \cite{de2012min,ravi2001approximation}. This highlights the need for faster and more scalable methods to design optimal communication topologies.

Quantum computing has emerged as a promising technology for tackling hard optimization problems, including those arising in optimal communication topology design. In contrast to classical bits, qubits can exist in superpositions and become entangled, enabling quantum algorithms that provably outperform the best-known classical methods for certain tasks, as in Shor's factoring and Grover's search algorithms \cite{nielsen2010quantum,shor1999polynomial,grover1996fast}. This has motivated a broader line of quantum algorithms for combinatorial optimization and Hamiltonian problems \cite{montanaro2016quantum}. At the same time, we are still in the noisy intermediate-scale quantum (NISQ) era: available devices have limited qubits and non-negligible gate errors, and large, fault-tolerant speedups on realistic control problems are not yet practical \cite{Preskill2018quantumcomputingin}. Consequently, most applications of quantum algorithms to classical optimization adopt a proof-of-concept viewpoint and rely on gate-level simulators, for example, in chemistry, portfolio optimization, or network design \cite{peruzzo2014variational,farhi2014quantum,ajagekar2019quantum}. The emphasis is on showing how mixed-integer or combinatorial models can be mapped to Hamiltonians and embedded into larger algorithmic or control architectures, so that the same formulations and circuits can be executed on real devices as larger and more reliable quantum processors become available.

Within this proof-of-concept landscape, several algorithmic paradigms have been proposed for ground-state and optimization tasks. Variational Quantum Eigensolver (VQE) \cite{peruzzo2014variational} and the Quantum Approximate Optimization Algorithm (QAOA) \cite{farhi2014quantum} are prominent examples: both use a parameterized quantum circuit (ansatz) evaluated on the quantum device, while a classical optimizer updates the parameters to minimize a cost function. These hybrid schemes are flexible and well-suited to NISQ hardware, but they introduce an outer classical optimization loop with step sizes, learning rates, and other hyperparameters that must be tuned, which can be burdensome when the same combinatorial problem must be solved repeatedly inside a control loop.

An alternative is the Quantum Imaginary Time Evolution (QITE) algorithm \cite{motta2020determining}. QITE can also be used to find low-energy solutions of an objective encoded as a Hamiltonian, but it is a fully quantum method: the ansatz parameters are updated directly through a quantum routine that approximates imaginary-time evolution, using quantum linear solvers or quantum differential equation solvers. In this way, QITE avoids the outer classical optimization loop that VQE and QAOA rely on. This leads to fewer hyperparameters to tune and a more streamlined training procedure, which is attractive for control and optimization applications.

Recent studies have begun exploring quantum computing and communication concepts in the context of MAS applications. For example, \cite{jenefa2024enhancing} introduces quantum protocols and entanglement to enhance distributed MAS environments, focusing primarily on secure synchronization and improved efficiency. The survey \cite{zhao2023quantum} discusses quantum multi-agent reinforcement learning, highlighting the potential of quantum tools in multi-agent learning scenarios; however, it does not address distributed control mechanisms. More recently, \cite{acha2025application} investigates quantum telecommunication for MASs, emphasizing quantum teleportation and wireless channels to mitigate communication delays. While these contributions lay useful foundations, they do not examine the application of quantum optimization algorithms, such as QITE, to design communication topologies.

\subsection{Contributions}
\label{subsec:contributions}

This paper presents a quantum optimization framework for consensus control in MASs, in which QITE dynamically optimizes the communication topology. The key idea is to decompose the MIQP problem into three blocks using an alternating direction method of multipliers (ADMM)-based decomposition: a convex classical subproblem, a pure binary subproblem suitable for quantum optimization, and a simple auxiliary update. The binary subproblem is then solved using QITE. The main contributions of this work are:
\begin{itemize}
    \item We formulate the communication topology design problem as an MIQP with practical cost terms and we enforce graph connectivity exactly using a standard flow-based formulation.
    \item We use a three-block ADMM splitting that relaxes the binary edge variables, introduces a binary proxy and an auxiliary continuous variable, and leads to (i) a convex QP subproblem for the relaxed topology and flows, (ii) a pure binary unconstrained subproblem in the binary proxy, and (iii) a closed-form update for the auxiliary variable.
    \item We map the binary unconstrained subproblem to a quadratic unconstrained binary optimization (QUBO) Hamiltonian and use QITE as a fully quantum solver to approximate its ground state.
\end{itemize}

\subsection{Preliminaries} \label{subsec:preliminaries}

 {Graph Theory \cite{mesbahi2010graph}:} A MAS is modeled as an undirected graph $\mathcal{G} = (\mathcal{V}, \mathcal{E})$, where $\mathcal{V} = \{1, 2, \dots, n\}$ is the set of nodes (agents), and $\mathcal{E} \subseteq \mathcal{V} \times \mathcal{V}$ is the set of communication links. A path is a sequence of edges and nodes between two distinct nodes with no duplicated edges or nodes. A graph is connected if a path exists between any pair of nodes. The adjacency matrix $A \in \mathbb{R}^{n \times n}$ is defined by $a_{ij} = 1$ if $(i,j) \in \mathcal{E}$, and $0$ otherwise. The degree matrix $D$ is diagonal with $d_{ii} = \sum_j a_{ij}$. The Laplacian matrix is given by $L = D - A$. The algebraic connectivity of the graph is the second smallest eigenvalue $\lambda_2$ of $L$, which reflects the graph's connectivity and affects consensus convergence.

 {Quantum Computing \cite{nielsen2010quantum}:} Quantum computing operates on qubits, which can exist in a superposition of pure states $|0\rangle=[1,0]^T$ and $|1\rangle=[0,1]^T$. A general single-qubit state is written as $|\psi\rangle = \alpha|0\rangle + \beta|1\rangle$, where $\alpha, \beta \in \mathbb{C}$ and $|\alpha|^2 + |\beta|^2 = 1$. Multiple qubits can be entangled, creating non-classical correlations across subsystems. Quantum gates are unitary matrices denoted by $U$, satisfying $U^\dagger U = I$, where $U^\dagger$ is the Hermitian conjugate of $U$. Applying a gate to a quantum state evolves it as $|\psi_1\rangle = U|\psi_0\rangle$. Typical quantum gates include single-qubit rotations (e.g., $R_x$, $R_y$, $R_z$) and two-qubit entangling gates such as CNOT. A quantum circuit is a sequence of such gates acting on one or more qubits. A final measurement collapses each qubit into a classical bit (0 or 1), with probabilities determined by the squared magnitudes of the amplitudes. Fig.~\ref{fig:qubit_measurement} illustrates the concept of a qubit in superposition and the probabilistic nature of quantum measurement. It also shows the Bloch sphere representation of a qubit.
\begin{figure}[!t]
    \centering
    \includegraphics[width=0.85\linewidth]{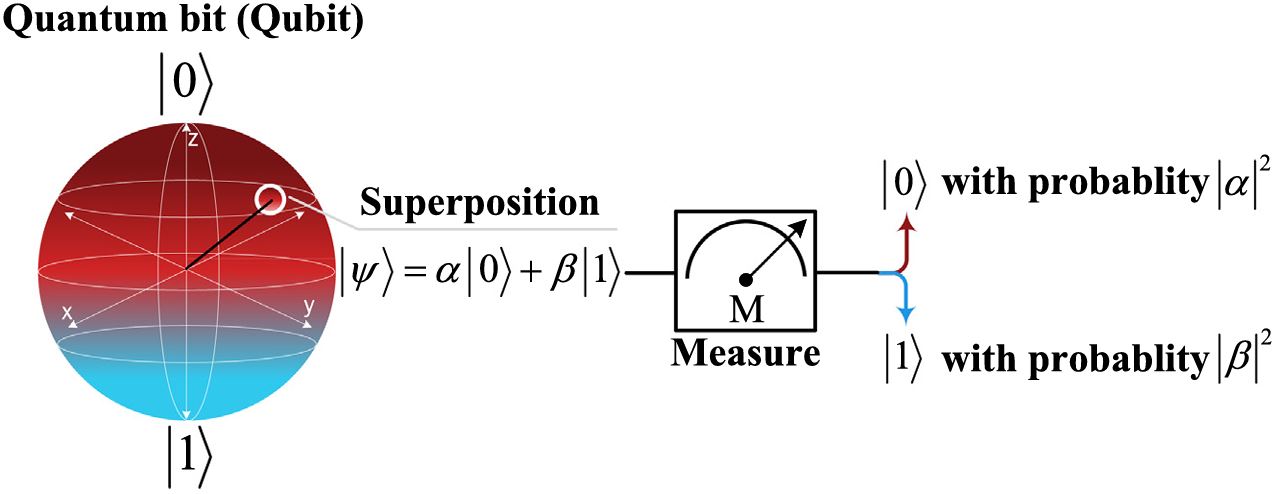}
    \caption{\footnotesize Visualization of a qubit in superposition and measurement outcomes.}
    \label{fig:qubit_measurement}
\end{figure}
Quantum circuits used for optimization can be broadly grouped into two types.  Variational circuits contain parameterized gates whose angles are tunable. Algorithms like QITE use parameterized quantum circuits and update the parameters directly via a quantum routine that approximates imaginary-time evolution. The goal is to prepare a quantum state that approximates the ground state of a Hamiltonian encoding the program.

\section{Problem Statement and Closed-Loop System} 
\label{sec:problem}

Here, we describe the MAS consensus dynamics, the assumptions regarding the communication network, and the communication topology optimization problem. We deliberately focus on linear first- and second-order consensus dynamics as a clear starting point. The optimization layer is combinatorial and does not depend on the linearity of the dynamics, so the same ideas can be applied to nonlinear models.

\subsection{Networked Consensus}
In this section, we recall standard consensus dynamics for first- and second-order multi-agent systems, and then state the assumptions under which the proposed time-varying, optimizer-generated graph guarantees consensus.

\subsubsection{First-Order Consensus Dynamics}
Consider a network of $n$ agents with first-order dynamics
\begin{align}\label{sys1}
 \dot{\boldsymbol{x}}(t) = \boldsymbol{u}(t),
\end{align}
where $\boldsymbol{x}(t) = [x_1(t), \dots, x_n(t)]^\top \in \mathbb{R}^n$ is the stacked vector of agent states and $\boldsymbol{u}(t) = [u_1(t), \dots, u_n(t)]^\top \in \mathbb{R}^n$ is the vector of control inputs. Consensus is achieved if
\[
\lim_{t \to \infty} |x_i(t) - x_j(t)| = 0, \quad \forall\, i,j \in \{1,\dots,n\}.
\]

As proposed in \cite{olfati2007consensus}, consensus in first-order MASs can be achieved using the classical distributed control law
\[
u_i(t) = -\sum_{j \in \mathcal{N}_i} a_{ij}\big(x_i(t) - x_j(t)\big),
\]
where $\mathcal{N}_i$ is the set of neighbors of agent $i$ and $a_{ij} \ge 0$ are the edge weights of the communication graph. In vector form, this becomes
\[
\boldsymbol{u}(t) = -L \boldsymbol{x}(t),
\]
where $L \in \mathbb{R}^{n \times n}$ is the graph Laplacian matrix. Substituting this controller into \eqref{sys1} yields the closed-loop system
\begin{align}\label{csys1_static}
    \dot{\boldsymbol{x}}(t) = -L \boldsymbol{x}(t).
\end{align}

If the graph is connected, then $L$ is positive semidefinite with a single zero eigenvalue associated with the eigenvector $\boldsymbol{1}$, and all other eigenvalues are strictly positive.

In the presence of a time-varying communication topology, the Laplacian becomes time-dependent, and the closed-loop dynamics are
\begin{align}\label{eq:csys1_tv}
  \dot{\boldsymbol{x}}(t) = -L(t)\, \boldsymbol{x}(t),
\end{align}
where $L(t)$ is the Laplacian of the communication graph at time $t$. In our framework, $L(t)$ is determined by the optimizer at each step and held constant until the next update.

\subsubsection{Second-Order Consensus Dynamics}
We now consider a network of $n$ agents with second-order dynamics. Each agent has a position $x_i(t)$ and a velocity $v_i(t)$, and the overall system is described by
\begin{align}
\dot{\boldsymbol{x}}(t) &= \boldsymbol{v}(t), \label{sys2_x}\\
\dot{\boldsymbol{v}}(t) &= \boldsymbol{u}(t), \label{sys2_v}
\end{align}
where $\boldsymbol{x}(t), \boldsymbol{v}(t), \boldsymbol{u}(t) \in \mathbb{R}^n$ are the stacked vectors of positions, velocities, and control inputs, respectively. Consensus in this case means that all agents reach the same position and move with the same velocity, i.e.,
\[
\lim_{t \to \infty} |x_i(t) - x_j(t)| = 0,\;
\lim_{t \to \infty} |v_i(t) - v_j(t)| = 0,
\quad \forall\, i,j.
\]

As shown in \cite{yu2009second}, second-order consensus can be achieved using the distributed controller
\[
u_i(t) = -\alpha \sum_{j \in \mathcal{N}_i} a_{ij}\big(x_i(t) - x_j(t)\big)
        -\beta \sum_{j \in \mathcal{N}_i} a_{ij}\big(v_i(t) - v_j(t)\big),
\]
where $\alpha,\beta > 0$ are control parameters and $a_{ij}$, $\mathcal{N}_i$ are as defined above. In vector form,
\[
\boldsymbol{u}(t) = -\alpha L \boldsymbol{x}(t) - \beta L \boldsymbol{v}(t).
\]
Substituting this into \eqref{sys2_x}–\eqref{sys2_v} yields the closed-loop second-order system
\begin{align}
\dot{\boldsymbol{x}}(t) &= \boldsymbol{v}(t), \label{csys2_static_x}\\
\dot{\boldsymbol{v}}(t) &= -\alpha L \boldsymbol{x}(t) - \beta L \boldsymbol{v}(t). \label{csys2_static_v}
\end{align}

With a time-varying topology, we obtain
\begin{align}
\dot{\boldsymbol{x}}(t) &= \boldsymbol{v}(t), \label{eq:csys2_tv_x}\\
\dot{\boldsymbol{v}}(t) &= -\alpha L(t)\, \boldsymbol{x}(t) - \beta L(t)\, \boldsymbol{v}(t), \label{eq:csys2_tv}
\end{align}
where $L(t)$ is again the Laplacian of the communication graph at time $t$. In the proposed framework, $L(t)$ is updated by the optimizer at discrete times and held constant in between.

\subsubsection{Assumptions and Consensus Guarantee}
The topology optimization algorithm described later produces, at each update time, a Laplacian $L(t)$ that is applied to the closed-loop consensus dynamics. In this subsection, we state the modeling assumptions on the communication layer and the resulting time-varying topology, and we recall the corresponding consensus guarantees.

\begin{assumption}[Ideal communication]
\label{ass:ideal_comm}
Communication links are lossless and delay-free, so the Laplacian $L(t)$ designed by the optimizer is exactly realized by the physical multi-agent system. Packet drops, noise, and time delays are not modeled in this work.
\end{assumption}

\begin{assumption}[Connectivity]
\label{ass:connectivity}
For all $t \ge 0$, $L(t)$ is the Laplacian of an undirected graph with nonnegative weights. The topology is piecewise constant, and there exists $\tau_{\min} > 0$ such that each topology is applied for at least $\tau_{\min}$ time units. Moreover, the graph is connected at all times, i.e., its algebraic connectivity satisfies $\lambda_2(L(t)) > 0$ for all $t$.
\end{assumption}

Under these assumptions, consensus is guaranteed for both first- and second-order dynamics.

\begin{theorem}
\label{thm:consensus}
Consider the first-order system \eqref{eq:csys1_tv} or the second-order system \eqref{eq:csys2_tv_x}–\eqref{eq:csys2_tv} under Assumptions~\ref{ass:ideal_comm}–\ref{ass:connectivity}. Then all agents achieve consensus. In particular, for all $i,j \in \{1,\dots,n\}$,
\[
\lim_{t\to\infty} |x_i(t) - x_j(t)| = 0,
\]
and, in the second-order case,
\[
\lim_{t\to\infty} |v_i(t) - v_j(t)| = 0.
\]
\end{theorem}

\begin{proof}
We sketch the main ideas; detailed proofs can be found in \cite{olfati2007consensus,yu2009second}.

\emph{First-order case:} Define the average state $\bar{x}(t) := \tfrac{1}{n}\boldsymbol{1}^\top \boldsymbol{x}(t)$ and the disagreement vector
\[
\boldsymbol{y}(t) := \boldsymbol{x}(t) - \bar{x}(t)\,\boldsymbol{1}
= \Big(I - \tfrac{1}{n}\boldsymbol{1}\boldsymbol{1}^\top\Big)\boldsymbol{x}(t).
\]
Using $L(t)\boldsymbol{1} = \boldsymbol{0}$ for all $t$, it is easy to see from \eqref{eq:csys1_tv} that $\dot{\bar{x}}(t) = 0$, so the average is invariant. The disagreement dynamics are
\[
\dot{\boldsymbol{y}}(t) = -L(t)\,\boldsymbol{y}(t),
\]
with $L(t)$ symmetric and positive semidefinite. By Assumption~\ref{ass:connectivity}, $\lambda_2(L(t)) > 0$ for all $t$, so along any nonzero disagreement direction the system has strictly negative instantaneous contraction rate. Since the topology is piecewise constant with a uniform dwell time, standard arguments for linear time-varying systems (see, e.g., \cite{olfati2007consensus}) imply that $\boldsymbol{y}(t) \to \boldsymbol{0}$ as $t \to \infty$. Therefore, all agents converge to the same value, equal to the initial average $\bar{x}(0)$, and consensus is achieved.

\emph{Second-order case:} Define the stacked state $\boldsymbol{\xi}(t) = [\boldsymbol{x}(t)^\top,\, \boldsymbol{v}(t)^\top]^\top$. In the disagreement coordinates obtained by projecting onto the subspace orthogonal to $\boldsymbol{1}$, the dynamics \eqref{eq:csys2_tv_x}–\eqref{eq:csys2_tv} decompose into $n-1$ independent second-order modes of the form
\[
\dot{\xi}_{1,i}(t) = \xi_{2,i}(t), \;
\dot{\xi}_{2,i}(t) = -\alpha \lambda_i(t)\, \xi_{1,i}(t) - \beta \lambda_i(t)\, \xi_{2,i}(t),
\]
where $\lambda_i(t)$, $i=2,\dots,n$, are the positive eigenvalues of $L(t)$ \cite{yu2009second}. For each mode, Assumption~\ref{ass:connectivity} implies $\lambda_i(t) > 0$, so the corresponding $2\times 2$ time-varying system matrix has eigenvalues with strictly negative real parts. Using Lyapunov arguments as in \cite{yu2009second}, together with the piecewise constant structure and dwell-time condition, one shows that all disagreement modes converge to zero. Thus, positions and velocities converge to a common value, resulting in second-order consensus.
\end{proof}

\begin{remark}
Stability and convergence depend only on the connectivity and nonnegativity conditions in Assumption~\ref{ass:connectivity}, not on global optimality of the topology with respect to any cost. The optimization layer is used to improve performance (e.g., communication cost), while feasibility and stability are enforced by the connectivity and degree constraints.
\end{remark}

\subsection{Communication Topology Optimization}
\label{subsec:topology_opt}

We now formalize the communication topology design problem that generates the Laplacian $L(t)$ used in the consensus dynamics. The goal is to choose, at discrete update times, a connected communication graph that balances communication cost, geometric distance, and fairness in the distribution of degrees across agents.

Let $0 = t_0 < t_1 < t_2 < \dots$ denote the discrete update times at which the communication topology is recomputed. At each $t_k$ we select an undirected graph \(\mathcal{G}_k\). The corresponding Laplacian matrix is denoted by $L_k$, and we set
\[
L(t) = L_k, \quad t \in [t_k, t_{k+1}).
\]
This piecewise-constant structure is consistent with Assumption~\ref{ass:connectivity} and ensures that the consensus dynamics see a constant topology between updates.

We start from the complete undirected candidate edge set
\[
\mathcal{E} = \big\{(i,j) \,\big|\, 1 \le i < j \le n\big\}.
\]
For each potential edge $(i,j) \in \mathcal{E}$ and update time $t_k$, we introduce a binary decision variable
\[
z_{ij}^k \in \{0,1\},
\]
where $z_{ij}^k = 1$ indicates that the link between agents $i$ and $j$ is active in $\mathcal{E}_k$, and $z_{ij}^k = 0$ means that the link is absent. When the time index is clear, we write $z_{ij}$ for simplicity.

At time $t_k$, each potential edge $(i,j)$ is assigned two basic cost components:
\begin{itemize}
    \item \emph{Communication cost} $c^{\text{comm}}_{ij} \ge 0$, which models hardware, bandwidth, or energy expenditure associated with maintaining the link $(i,j)$;
    \item \emph{Distance cost} $d_{ij} \ge 0$, typically based on the spatial distance between agents, e.g.,
    \[
        d_{ij} = \big\|x_i(t_k) - x_j(t_k)\big\|_2,
    \]
    so that long-range links are penalized more heavily than short-range ones.
\end{itemize}
We combine these into a single linear weight
\begin{equation}
    w_{ij} :=  c^{\text{comm}}_{ij}
            +  d_{ij},
    \quad \forall (i,j) \in \mathcal{E}.
    \label{eq:edge_weight}
\end{equation}

The total linear cost of a topology is then
\begin{equation}
    J_{\text{lin}}(z) = \sum_{(i,j)\in\mathcal{E}} w_{ij} z_{ij}.
    \label{eq:topology_linear_cost}
\end{equation}

To avoid centralized topologies in which a few nodes act as large hubs, we include a quadratic penalty on the degrees of the nodes. Under a topology $z$, the degree of node $i$ is
\begin{equation}
    \deg(i) = \sum_{j \neq i} z_{\min\{i,j\},\max\{i,j\}},
    \label{eq:degree_def}
\end{equation}
and we add the quadratic term
\begin{equation}
    J_{\text{deg}}(z)
    = \sum_{i=1}^n \kappa_i \big(\deg(i)\big)^2,
    \label{eq:degree_quadratic_cost}
\end{equation}
where $\kappa_i \ge 0$ determines how strongly we penalize high degree at node $i$. This term encourages the optimizer to spread edges more evenly across nodes.

In addition to the quadratic penalty, we impose a hard upper bound on the degree of every node:
\begin{equation}
    \deg(i) \le \gamma, 
    \quad \forall i \in \mathcal{V},
    \label{eq:degree_constraint}
\end{equation}
where $\gamma$ is a prescribed maximum degree. Constraint~\eqref{eq:degree_constraint} limits the worst-case communication load on any single agent, while the quadratic term \eqref{eq:degree_quadratic_cost} softly discourages uneven degree distributions within that feasible set.

The consensus guarantees in Theorem~\ref{thm:consensus} require that the communication graph be connected at all times. To enforce connectivity exactly, we use a standard flow-based formulation that introduces artificial flows over the graph.

Fix an arbitrary root node $r \in \mathcal{V}$ and define the directed arc set
\[
\mathcal{A} = \big\{(i,j),(j,i) \,\big|\, (i,j)\in\mathcal{E}\big\}.
\]
For each directed arc $(i,j) \in \mathcal{A}$, we introduce a nonnegative continuous flow variable $f_{ij} \ge 0$. Intuitively, one unit of artificial flow must be delivered from the root to each other node, and flow is only allowed to use active edges.

We impose capacity constraints
\begin{equation}
    0 \le f_{ij} \le (n-1)\, z_{\min\{i,j\},\max\{i,j\}},
    \quad \forall (i,j) \in \mathcal{A},
    \label{eq:flow_capacity}
\end{equation}
and the flow conservation constraints
\begin{align}
    \sum_{j:(j,k)\in\mathcal{A}} f_{jk}
    - \sum_{j:(k,j)\in\mathcal{A}} f_{kj}
    &= 1, && \forall k \in \mathcal{V}\setminus\{r\},
    \label{eq:flow_nonroot}\\
    \sum_{j:(r,j)\in\mathcal{A}} f_{rj}
    - \sum_{j:(j,r)\in\mathcal{A}} f_{jr}
    &= n-1, && \text{for the root } r.
    \label{eq:flow_root}
\end{align}
If a solution satisfies \eqref{eq:flow_capacity}–\eqref{eq:flow_root}, then there exists a spanning tree rooted at $r$, and the undirected graph induced by $\{(i,j)\in\mathcal{E} : z_{ij}=1\}$ is connected.

Collecting the objective and constraints, the communication topology design problem at time $t_k$ is written as the MIQP
\begin{subequations}
\label{eq:MIQP_topology}
\begin{align}
    \min_{z,f} \;
    & J(z) := J_{\text{lin}}(z) + J_{\text{deg}}(z)
      \label{eq:MIQP_obj}\\[1mm]
    \text{s.t.} \;
    & z_{ij} \in \{0,1\},
      \quad \forall (i,j) \in \mathcal{E},\\
    & f_{ij} \ge 0,
      \quad \forall (i,j) \in \mathcal{A},\\
    & \deg(i) \le \gamma, 
      \quad \forall i \in \mathcal{V},
      \label{eq:MIQP_deg}\\
    & 0 \le f_{ij} \le (n-1)\,
      z_{\min\{i,j\},\max\{i,j\}},
      \nonumber\\
    & \hphantom{0 \le f_{ij} \le (n-1)\,}
      \forall (i,j) \in \mathcal{A},
      \label{eq:MIQP_cap}\\
    & \sum_{j} f_{jk} - \sum_{j} f_{kj} = 1,
      \nonumber\\
    & \hphantom{\sum_{j} f_{jk} - \sum_{j} f_{kj} = 1,}
      \forall k \in \mathcal{V}\setminus\{r\},
      \label{eq:MIQP_flow_nonroot}\\
    & \sum_{j} f_{rj} - \sum_{j} f_{jr} = n-1.
      \label{eq:MIQP_flow_root}
\end{align}
\end{subequations}

The binary variables $z_{ij}$ select the active edges, the continuous flow variables $f_{ij}$ certify connectivity via the constraints \eqref{eq:MIQP_cap}–\eqref{eq:MIQP_flow_root}, and the degree-related terms \eqref{eq:MIQP_deg}–\eqref{eq:MIQP_obj} enforce sparsity. Solving \eqref{eq:MIQP_topology} at each update time $t_k$ produces a connected topology and a corresponding $L_k$ that is used in the consensus dynamics.

\begin{assumption}\label{ass:gamma_feasible}
The maximum degree satisfies $1 \le \gamma \le n-1$, and for $n \ge 3$ we require $\gamma \ge 2$ so that a connected graph with bounded degree exists.
\end{assumption}

\begin{lemma}\label{lem:feasibility}
Under Assumption~\ref{ass:gamma_feasible}, the MIQP
\eqref{eq:MIQP_topology} has a nonempty feasible set.
\end{lemma}

\begin{proof}
We construct an explicit feasible solution $(z,f)$.

Step 1 (choice of a connected graph with bounded degree). For $n = 1$ there are no edges and the problem is trivially feasible. For $n = 2$, a single edge between the two nodes gives a connected
graph with degrees $\deg(1)=\deg(2)=1$, which satisfies $1 \le \gamma \le n-1 = 1$ by Assumption~\ref{ass:gamma_feasible}. Thus the degree constraints hold.

Now consider $n \ge 3$. By Assumption~\ref{ass:gamma_feasible},
we have $2 \le \gamma \le n-1$. Define a path graph on the node set
$\mathcal{V} = \{1,\dots,n\}$ by setting
\[
z_{i,i+1} = 1,\quad i = 1,\dots,n-1,
\]
and $z_{ij} = 0$ for all other pairs. The resulting undirected graph
is connected and the node degrees satisfy
\[
\deg(1) = \deg(n) = 1,\quad \deg(i) = 2,\; i=2,\dots,n-1.
\]
Therefore, $\deg(i) \le 2 \le \gamma$ for all $i$, and the degree
constraints \eqref{eq:MIQP_deg} are satisfied. Since we never exceed
degree $2$, the upper bound $\gamma \le n-1$ is automatically
respected.

Step 2 (flows that satisfy the connectivity constraints).
Fix the root node $r = 1$ without loss of generality. On the undirected
path, there is a unique simple path from $r$ to every node $k$.
Define the directed flow variables on the arc set $\mathcal{A}$ as
follows. For each edge $(i,i+1)$ with $i = 1,\dots,n-1$, set
\[
f_{i,i+1} = n-i,\quad f_{i+1,i} = 0,
\]
and $f_{ij} = 0$ for all other directed arcs $(i,j) \in \mathcal{A}$.
The quantity $f_{i,i+1}$ represents how many downstream nodes must
receive one unit of artificial flow from the root.

By construction, $f_{ij} \ge 0$ for all $(i,j)$, and since
$n-i \le n-1$ and $z_{i,i+1}=1$, the capacity constraints
\eqref{eq:MIQP_cap} hold:
\[
0 \le f_{i,i+1} = n-i \le n-1 = (n-1)\, z_{i,i+1}.
\]
For each non-root node $k>1$, the net inflow is
\[
\sum_{j} f_{jk} - \sum_{j} f_{kj} = 1,
\]
and for the root $r=1$ the net outflow equals $n-1$. Thus the flow
conservation constraints \eqref{eq:MIQP_flow_nonroot}–\eqref{eq:MIQP_flow_root}
are satisfied.

We have therefore exhibited $(z,f)$ that satisfies all constraints of
\eqref{eq:MIQP_topology}, so the feasible set is nonempty.
\end{proof}

\section{Three-Block ADMM Decomposition}
\label{sec:three_block_admm}

The MIQP in \eqref{eq:MIQP_topology} gives an exact formulation of the topology design problem at each update time. However, it is difficult to solve this mixed-integer problem to optimality at every step when the network is large. In this subsection, we describe how to decompose the problem into three simpler blocks using the ADMM. The goal is to keep the connectivity and degree constraints in a classical convex block, to move the binary decisions into a pure binary block, and to use a third block to couple them.
\subsection{Three-Block Formulation}

We keep the binary edge variables of the MIQP and introduce two extra copies of them. For each undirected edge $(i,j)\in\mathcal{E}$ we define three variables:
\begin{itemize}
    \item a relaxed edge variable $z_{ij} \in [0,1]$, which appears in the cost and in the constraints of \eqref{eq:MIQP_topology};
    \item a binary proxy variable $r_{ij} \in \{0,1\}$, which will be the input to the binary subproblem;
    \item an auxiliary variable $s_{ij} \in \mathbb{R}$, which helps to enforce agreement between $z_{ij}$ and $r_{ij}$.
\end{itemize}

If we stack all edges into vectors $z,r,s\in\mathbb{R}^m$ with $m = |\mathcal{E}|$, we connect them with the linear constraint
\begin{equation}
    z - r + s = 0,
    \label{eq:three_block_consensus}
\end{equation}
which holds elementwise. In the original MIQP we would require $z \in \{0,1\}^m$. In the three-block scheme, we relax this to $z \in [0,1]^m$, keep $r$ binary, and use the constraint \eqref{eq:three_block_consensus} with penalties and dual variables to keep $z$ and $r$ close to each other.

The constraints of the MIQP, namely the degree constraint \eqref{eq:MIQP_deg}, the capacity constraint \eqref{eq:MIQP_cap}, and the flow conservation constraints \eqref{eq:MIQP_flow_nonroot}–\eqref{eq:MIQP_flow_root}, together with the box constraint $0 \le z \le 1$, define a convex feasible set for $(z,f)$ once integrality is removed.

To apply ADMM, we introduce a dual variable $\lambda \in \mathbb{R}^m$ for the constraint \eqref{eq:three_block_consensus} and add a quadratic penalty. We also add a small regularization on $s$. Let $\rho > 0$ and $\beta \ge 0$ be fixed parameters. We use an indicator function for the convex constraints of the MIQP:
\begin{align}
\iota_{\mathcal{C}}(z,f) =
\begin{cases}
0, & \text{if \eqref{eq:MIQP_deg}--\eqref{eq:MIQP_flow_root} and } 0 \le z \le 1 \text{ hold},\\[1mm]
+\infty, & \text{otherwise.}
\end{cases}
\end{align}

The augmented Lagrangian associated with \eqref{eq:MIQP_topology} and \eqref{eq:three_block_consensus} is
\begin{align}
\mathcal{L}(z,f,r,s,\lambda)
=& J(z) + \iota_{\mathcal{C}}(z,f)  + \lambda^\top (z - r + s)
      \nonumber\\
&+ \frac{\rho}{2} \|z - r + s\|_2^2
      + \frac{\beta}{2} \|s\|_2^2,
\label{eq:aug_lagrangian_three_block}
\end{align}
where $J(z)$ is the MIQP objective in \eqref{eq:MIQP_obj}. The ADMM method minimizes $\mathcal{L}$ with respect to each block $(z,f)$, $r$, and $s$ in turn, and then updates the dual variable $\lambda$.

In Block 1, we update the relaxed topology $z$ and the flows $f$. At iteration $k$, given $r^{(k-1)}$, $s^{(k-1)}$, and $\lambda^{(k-1)}$, we solve
\begin{align}
(z^{(k)}, f^{(k)})
= \arg\min_{z,f} \quad &
\mathcal{L}(z,f, r^{(k-1)}, s^{(k-1)}, \lambda^{(k-1)})
\label{eq:block1_full}
\end{align}
subject to the convex constraints encoded in $\iota_{\mathcal{C}}(z,f)$.

Since the integrality constraint on $z$ is relaxed and all constraints in $\mathcal{C}$ are linear, this is a convex quadratic optimization problem. It can be solved by standard classical solvers. Intuitively, Block 1 chooses a connected topology and flows that minimize the MIQP cost plus a quadratic term that keeps $z$ close to $r^{(k-1)} - s^{(k-1)}$.

In Block 2, we update the binary proxy $r$. At iteration $k$, we fix $z^{(k)}$, $f^{(k)}$, $s^{(k-1)}$, and $\lambda^{(k-1)}$, and solve
\begin{equation}
    r^{(k)} =
    \arg\min_{r \in \{0,1\}^m}
    \mathcal{L}(z^{(k)}, f^{(k)}, r, s^{(k-1)}, \lambda^{(k-1)}).
    \label{eq:block2_r_problem}
\end{equation}

The terms in $\mathcal{L}$ that depend on $r$ are
\[
    -\lambda^{(k-1)\top} r
    + \frac{\rho}{2} \|z^{(k)} - r + s^{(k-1)}\|_2^2.
\]
If we expand the squared norm, we obtain a quadratic function in $r$ with no constraints other than $r \in \{0,1\}^m$. This is a QUBO problem.

In Block 3, we update the auxiliary vector $s$. At iteration $k$, we fix $z^{(k)}$, $r^{(k)}$, and $\lambda^{(k-1)}$ and solve
\begin{equation}
    s^{(k)} =
    \arg\min_{s \in \mathbb{R}^m}
    \mathcal{L}(z^{(k)}, f^{(k)}, r^{(k)}, s, \lambda^{(k-1)}).
    \label{eq:block3_s_problem}
\end{equation}

For each component $e = 1,\dots,m$, the dependence on $s_e$ is quadratic:
\[
\lambda_e^{(k-1)} (z_e^{(k)} - r_e^{(k)} + s_e)
+ \frac{\rho}{2} (z_e^{(k)} - r_e^{(k)} + s_e)^2
+ \frac{\beta}{2} s_e^2.
\]
Taking the derivative with respect to $s_e$ and setting it to zero gives
\[
\lambda_e^{(k-1)}
+ \rho \big(z_e^{(k)} - r_e^{(k)} + s_e^{(k)}\big)
+ \beta s_e^{(k)} = 0.
\]
Therefore
\begin{equation}
    s_e^{(k)}
    = -\frac{\lambda_e^{(k-1)} + \rho \big(z_e^{(k)} - r_e^{(k)}\big)}
            {\rho + \beta},
    \quad e = 1,\dots,m.
    \label{eq:block3_closed_form}
\end{equation}
This update is fully explicit and independent across edges, so Block 3 is computationally cheap.

After updating $(z^{(k)}, f^{(k)}, r^{(k)}, s^{(k)})$, we update the dual variables. The residual of the coupling constraint \eqref{eq:three_block_consensus} at iteration $k$ is
\begin{align}\label{residual}
    d^{(k)} := z^{(k)} - r^{(k)} + s^{(k)}.
\end{align}

The dual variables are updated by
\begin{equation}
    \lambda^{(k)} =
    \lambda^{(k-1)} + \rho\, d^{(k)}.
    \label{eq:dual_update_three_block}
\end{equation}
If the method converges, then $d^{(k)}$ tends to zero, so $z$, $r$, and $s$ agree asymptotically.

Algorithm~\ref{alg:admm_topology} summarizes the three-block ADMM scheme for a single topology update time. In the closed-loop setting, this algorithm is called at each $t_k$, and the resulting edge vector $z^{(K)}$ is used to build the Laplacian $L_k$.

\begin{algorithm}[t!]
  \caption{Three-block ADMM topology solver at update time $t_k$}
  \label{alg:admm_topology}
  \begin{algorithmic}[1]
    \State \textbf{Initialize:} $(z^{(0)}, f^{(0)}, r^{(0)}, s^{(0)}, \lambda^{(0)})$, penalty parameters $\rho>0$, $\beta\ge 0$, and stopping tolerances. Set $k \gets 1$.
    \While{stopping criterion not met}
      \State \textbf{Block 1:}
      Compute $(z^{(k)}, f^{(k)})$ by solving\eqref{eq:block1_full}.
      \State \textbf{Block 2:}
      Update $r^{(k)}$ by solving \eqref{eq:block2_r_problem}.
      \State \textbf{Block 3:}
      Update $s^{(k)}$ componentwise using \eqref{eq:block3_closed_form}.
      \State \textbf{Dual update:}
      Update according to \eqref{eq:dual_update_three_block}.
      \State $k \gets k+1$.
    \EndWhile
    \State \textbf{Return:} final edge vector $z^{(k)}$ and Laplacian $L_k$ constructed from the active edges.
  \end{algorithmic}
\end{algorithm}

\subsection{Convergence and Stability of the Three-Block ADMM}
\label{subsec:conv_admm}

Here, we study the convergence and the algorithmic stability of the three-block ADMM scheme in Algorithm~\ref{alg:admm_topology}. Under a set of regularity assumptions stated below, we show that the consensus residual converges to zero and that any limit point of the ADMM iterates satisfies the first-order optimality conditions of the augmented Lagrangian. We focus on the iteration of the optimization layer, not on the physical stability of the multi-agent consensus dynamics, which was already discussed in Theorem~\ref{thm:consensus}.

We recall that for each undirected edge, we introduced three variables: the relaxed edge variable, the binary proxy, and the auxiliary variable. If we stack all edges in a fixed order, we obtain the vectors $z,r,s \in \mathbb{R}^m$, where $m = |\mathcal{E}|$. The coupling between these three vectors is expressed by the consensus constraint in \eqref{eq:three_block_consensus}. The augmented Lagrangian that we use in the three-block ADMM is given in \eqref{eq:aug_lagrangian_three_block}. At iteration $k$, Algorithm~\ref{alg:admm_topology} first updates the relaxed topology and flow variables $(z^{(k)},f^{(k)})$, then updates the binary proxy $r^{(k)}$, then updates the auxiliary variable $s^{(k)}$, and finally updates the dual variable $\lambda^{(k)}$.

To study the behavior of this iteration, we treat the ADMM scheme as a discrete-time dynamical system. We denote the stacked iterate at step $k$ as
\[
    \Xi^{(k)} := \big(z^{(k)}, f^{(k)}, r^{(k)}, s^{(k)}, \lambda^{(k)}\big).
\]
Then Algorithm~\ref{alg:admm_topology} defines a recursion of the form
\[
    \Xi^{(k+1)} = \mathcal{T}\big(\Xi^{(k)}\big),
\]
where $\mathcal{T}$ is the operator that represents one full ADMM sweep. The dual update step can be written as
\begin{equation}
    \lambda^{(k+1)} = \lambda^{(k)} + \rho\, d^{(k+1)},
    \label{eq:dual_update_residual_form_conv}
\end{equation}
where $d^{(k+1)}$ is the consensus residual at iteration $k+1$ defined in \eqref{residual}. This relation shows that changes in the dual variable are directly driven by the residual of the consensus constraint.

\begin{assumption}[Bounded iterates]
\label{ass:admm_bounded}
The sequence of iterates
\[
    \Xi^{(k)} := \big(z^{(k)}, f^{(k)}, r^{(k)}, s^{(k)}, \lambda^{(k)}\big)
\]
generated by Algorithm~\ref{alg:admm_topology} is bounded.
In other words, there exists a compact set such that all these iterates remain inside it for all $k\ge 0$.
\end{assumption}

This is a standard technical assumption in nonconvex ADMM analysis and can be enforced in practice, e.g., by projecting the dual variable onto a large box.

\begin{assumption}[Exact block updates]
\label{ass:admm_exact}
At every iteration $k$, each block subproblem in Algorithm~\ref{alg:admm_topology} is solved exactly. 
Block~1 returns an exact minimizer $(z^{(k)},f^{(k)})$ of the convex QP in \eqref{eq:block1_full} subject to the constraints collected in $\mathcal{C}$. 
Block~2 returns an exact minimizer $r^{(k)}$ of the QUBO problem \eqref{eq:block2_r_problem} over $\{0,1\}^m$. 
Block~3 updates $s^{(k)}$ using the explicit formula in \eqref{eq:block3_closed_form}.
\end{assumption}

\begin{assumption}[Strong convexity of Block~1]
\label{ass:admm_strong_convex}
For any fixed $(r,s,\lambda)$, the objective of the Block~1 subproblem is strongly convex in the variables $(z,f)$ over the feasible set defined by $\mathcal{C}$. 
This strong convexity is induced, for example, by the quadratic degree penalty $J_{\deg}(z)$ together with the quadratic penalty $\tfrac{\rho}{2}\|z-r+s\|_2^2$ in \eqref{eq:aug_lagrangian_three_block}, and possibly by a small convex regularization in $f$. For example, one can add a small term $\frac{\epsilon}{2}\|f\|_2^2$ with $\epsilon>0$ to the objective of Block~1 to ensure strong convexity with respect to $f$.

\end{assumption}

\begin{assumption}[Regularity]
\label{ass:admm_regularity}
The augmented Lagrangian $\mathcal{L}$ in \eqref{eq:aug_lagrangian_three_block} is proper, lower semicontinuous, and bounded from below on the domain defined by the constraints in $\mathcal{C}$ and by $r\in\{0,1\}^m$. 
Moreover, $\mathcal{L}$ satisfies the Kurdyka--\L ojasiewicz property. 
This holds, for example, because $J(z)$ is a sum of linear and convex quadratic terms, $\iota_{\mathcal{C}}$ is the indicator of a polyhedral set, the remaining terms in \eqref{eq:aug_lagrangian_three_block} are quadratic in $(z,r,s,\lambda)$, and the overall function is semialgebraic in $(z,f,r,s,\lambda)$.
\end{assumption}

Under these assumptions, we can use a Lyapunov candidate to analyze the behavior of the three-block ADMM. The idea is to combine the augmented Lagrangian and a quadratic penalty on the residual into a single scalar function that does not increase along the iterations and that is bounded from below. We define
\begin{equation}
    V^{(k)} :=
    \mathcal{L}\big(z^{(k)},f^{(k)},r^{(k)},s^{(k)},\lambda^{(k)}\big)
    + \frac{\kappa}{2}\,\big\|d^{(k)}\big\|_2^2,
    \label{eq:Lyapunov_candidate_conv}
\end{equation}
where $\kappa>0$ is a tuning parameter. The first term in $V^{(k)}$ measures the objective plus the constraint penalties at the current iterate. The second term penalizes the violation of the consensus constraint. If the algorithm converges to a point where the consensus constraint is satisfied, then both terms should approach finite limits.

To study how $V^{(k)}$ changes from one iteration to the next, we write
\begin{align}
V^{(k+1)} - V^{(k)}
&=
\mathcal{L}\big(z^{(k+1)},f^{(k+1)},r^{(k+1)},s^{(k+1)},\lambda^{(k+1)}\big)
\nonumber\\
&\quad - \mathcal{L}\big(z^{(k)},f^{(k)},r^{(k)},s^{(k)},\lambda^{(k)}\big)
\nonumber\\
&\quad + \frac{\kappa}{2}\left(\big\|d^{(k+1)}\big\|_2^2 - \big\|d^{(k)}\big\|_2^2\right).
\label{eq:V_diff_step1_conv}
\end{align}
For notational convenience, we denote by
\[
    \mathcal{L}^{k}(\lambda) :=
    \mathcal{L}\big(z^{(k)},f^{(k)},r^{(k)},s^{(k)},\lambda\big)
\]
the augmented Lagrangian evaluated at the primal variables of iteration $k$ and at a generic dual vector $\lambda$. Then we can rewrite \eqref{eq:V_diff_step1_conv} as
\begin{align}
V^{(k+1)} - V^{(k)}
&=
\mathcal{L}^{k+1}(\lambda^{(k+1)}) - \mathcal{L}^{k}(\lambda^{(k)})
\nonumber\\
&\quad + \frac{\kappa}{2}\left(\big\|d^{(k+1)}\big\|_2^2 - \big\|d^{(k)}\big\|_2^2\right).
\label{eq:V_diff_step2_conv}
\end{align}

We now split the change in the augmented Lagrangian into a dual part and a primal part. We add and subtract $\mathcal{L}^{k+1}(\lambda^{(k)})$ inside the first difference:
\begin{align}
\mathcal{L}^{k+1}(\lambda^{(k+1)}) - \mathcal{L}^{k}(\lambda^{(k)})
&=
\big[\mathcal{L}^{k+1}(\lambda^{(k+1)}) - \mathcal{L}^{k+1}(\lambda^{(k)})\big]
\nonumber\\
&\quad + \big[\mathcal{L}^{k+1}(\lambda^{(k)}) - \mathcal{L}^{k}(\lambda^{(k)})\big].
\label{eq:L_split_conv}
\end{align}
The first bracket corresponds to the effect of updating the dual variable from $\lambda^{(k)}$ to $\lambda^{(k+1)}$ while the primal variables are fixed at $(z^{(k+1)},f^{(k+1)},r^{(k+1)},s^{(k+1)})$. The second bracket corresponds to the effect of updating the primal variables from iteration $k$ to iteration $k+1$ while the dual variable remains equal to $\lambda^{(k)}$.

We first consider the dual part. In the augmented Lagrangian in \eqref{eq:aug_lagrangian_three_block}, the dual variable appears only in the linear term $\lambda^\top(z-r+s)$. Therefore, for any fixed $(z,f,r,s)$ and any two dual vectors $\lambda$ and $\tilde{\lambda}$, we have
\[
    \mathcal{L}(z,f,r,s,\tilde{\lambda}) - \mathcal{L}(z,f,r,s,\lambda)
    = (\tilde{\lambda} - \lambda)^\top (z - r + s).
\]
If we apply this identity with $\lambda=\lambda^{(k)}$, $\tilde{\lambda}=\lambda^{(k+1)}$ and the primal variables set to their $(k+1)$ values, then $z-r+s$ is equal to $d^{(k+1)}$ and we obtain
\begin{align}
\mathcal{L}^{k+1}(\lambda^{(k+1)}) - \mathcal{L}^{k+1}(\lambda^{(k)})
&=
\big(\lambda^{(k+1)} - \lambda^{(k)}\big)^\top d^{(k+1)}.
\label{eq:dual_change_step_conv}
\end{align}
Using the dual update \eqref{eq:dual_update_residual_form_conv}, we have $\lambda^{(k+1)} - \lambda^{(k)} = \rho\, d^{(k+1)}$, and therefore
\begin{align}
\mathcal{L}^{k+1}(\lambda^{(k+1)}) - \mathcal{L}^{k+1}(\lambda^{(k)})
&=
\rho\, \big\|d^{(k+1)}\big\|_2^2.
\label{eq:dual_change_final_conv}
\end{align}

We now consider the primal part, that is the second bracket in \eqref{eq:L_split_conv}. This term is
\[
    \mathcal{L}^{k+1}(\lambda^{(k)}) - \mathcal{L}^{k}(\lambda^{(k)}),
\]
and it measures the change in the augmented Lagrangian due to the updates of $(z,f)$, $r$ and $s$ at fixed dual variable $\lambda^{(k)}$. We follow the three primal blocks in Algorithm~\ref{alg:admm_topology}.

In Block~1, the pair $(z^{(k+1)},f^{(k+1)})$ is obtained by minimizing $\mathcal{L}$ with respect to $(z,f)$ under the convex constraints in $\mathcal{C}$, with $(r,s,\lambda)$ fixed at $(r^{(k)},s^{(k)},\lambda^{(k)})$. By strong convexity of the Block~1 objective in $(z,f)$, there exists a constant $\sigma_1>0$ such that
\begin{align}
\mathcal{L}\big(z^{(k+1)},&f^{(k+1)},r^{(k)},s^{(k)},\lambda^{(k)}\big)
\nonumber\\&\le
\mathcal{L}\big(z^{(k)},f^{(k)},r^{(k)},s^{(k)},\lambda^{(k)}\big)
\nonumber\\
&\quad - \sigma_1 \big\|(z^{(k+1)},f^{(k+1)}) - (z^{(k)},f^{(k)})\big\|_2^2.
\label{eq:block1_descent_conv}
\end{align}
This inequality means that, if $(z,f)$ moves, then the augmented Lagrangian decreases by at least a quadratic amount in the step size of $(z,f)$.

In Block~2, the vector $r^{(k+1)}$ is obtained by minimizing $\mathcal{L}$ with respect to $r$ in the set $\{0,1\}^m$ with $(z,f,s,\lambda)$ fixed at $(z^{(k+1)},f^{(k+1)},s^{(k)},\lambda^{(k)})$. Since $r^{(k+1)}$ is a minimizer of this subproblem, we have
\begin{align}
\mathcal{L}\big(z^{(k+1)},&f^{(k+1)},r^{(k+1)},s^{(k)},\lambda^{(k)}\big)
\nonumber\\&\le
\mathcal{L}\big(z^{(k+1)},f^{(k+1)},r^{(k)},s^{(k)},\lambda^{(k)}\big).
\label{eq:block2_descent_conv}
\end{align}
Here, we only know that the augmented Lagrangian does not increase in Block~2.

In Block~3, the vector $s^{(k+1)}$ is obtained by minimizing $\mathcal{L}$ with respect to $s$ in $\mathbb{R}^m$ with $(z,f,r,\lambda)$ fixed at $(z^{(k+1)},f^{(k+1)},r^{(k+1)},\lambda^{(k)})$. This minimization has the closed-form solution in \eqref{eq:block3_closed_form}. The dependence of $\mathcal{L}$ on $s$ is strictly convex and quadratic, so there exists a constant $\sigma_2>0$ such that
\begin{align}
\mathcal{L}\big(z^{(k+1)},&f^{(k+1)},r^{(k+1)},s^{(k+1)},\lambda^{(k)}\big)
\nonumber\\&\le
\mathcal{L}\big(z^{(k+1)},f^{(k+1)},r^{(k+1)},s^{(k)},\lambda^{(k)}\big)
\nonumber\\
&\quad - \sigma_2 \big\|s^{(k+1)} - s^{(k)}\big\|_2^2.
\label{eq:block3_descent_conv}
\end{align}

If we combine \eqref{eq:block1_descent_conv}, \eqref{eq:block2_descent_conv} and \eqref{eq:block3_descent_conv}, and we identify
\[
    \mathcal{L}^{k}(\lambda^{(k)}) =
    \mathcal{L}\big(z^{(k)},f^{(k)},r^{(k)},s^{(k)},\lambda^{(k)}\big),
\]
and
\[
    \mathcal{L}^{k+1}(\lambda^{(k)}) =
    \mathcal{L}\big(z^{(k+1)},f^{(k+1)},r^{(k+1)},s^{(k+1)},\lambda^{(k)}\big),
\]
we obtain the bound
\begin{align}
\mathcal{L}^{k+1}(\lambda^{(k)}) - &\mathcal{L}^{k}(\lambda^{(k)})
\le
\nonumber\\&- \sigma_1 \big\|(z^{(k+1)},f^{(k+1)}) - (z^{(k)},f^{(k)})\big\|_2^2
\nonumber\\
& - \sigma_2 \big\|s^{(k+1)} - s^{(k)}\big\|_2^2.
\label{eq:primal_change_bound_conv}
\end{align}

Now we substitute \eqref{eq:L_split_conv}, \eqref{eq:dual_change_final_conv} and \eqref{eq:primal_change_bound_conv} into \eqref{eq:V_diff_step2_conv}. We get
\begin{align}
V^{(k+1)} - V^{(k)}
&\le
\rho\, \big\|d^{(k+1)}\big\|_2^2
\nonumber\\
&\quad - \sigma_1 \big\|(z^{(k+1)},f^{(k+1)}) - (z^{(k)},f^{(k)})\big\|_2^2
\nonumber\\
&\quad - \sigma_2 \big\|s^{(k+1)} - s^{(k)}\big\|_2^2
\nonumber\\
&\quad + \frac{\kappa}{2}\left(\big\|d^{(k+1)}\big\|_2^2 - \big\|d^{(k)}\big\|_2^2\right).
\label{eq:V_diff_step3_conv}
\end{align}
Collecting the terms with $d^{(k+1)}$ and $d^{(k)}$, we can rewrite this as
\begin{align}
V^{(k+1)} - V^{(k)}
&\le
\left(\rho + \frac{\kappa}{2}\right)\big\|d^{(k+1)}\big\|_2^2
- \frac{\kappa}{2}\big\|d^{(k)}\big\|_2^2
\nonumber\\
&\quad - \sigma_1 \big\|(z^{(k+1)},f^{(k+1)}) - (z^{(k)},f^{(k)})\big\|_2^2
\nonumber\\
&\quad - \sigma_2 \big\|s^{(k+1)} - s^{(k)}\big\|_2^2.
\label{eq:V_diff_final_conv}
\end{align}

Inequality \eqref{eq:V_diff_final_conv} shows that the Lyapunov candidate $V^{(k)}$ can only increase through the terms that depend on the residuals, while the changes in $(z,f)$ and in $s$ always decrease the augmented Lagrangian by at least a quadratic amount. Combined with the lower boundedness and regularity of $\mathcal{L}$ in Assumption~\ref{ass:admm_regularity}, this type of inequality is standard in the convergence analysis of nonconvex ADMM schemes, see for example \cite{hong2016convergence,wang2019global,guo2017convergence}.

\begin{theorem}
\label{thm:admm_convergence}
Suppose that Assumptions~\ref{ass:admm_bounded}--\ref{ass:admm_regularity} hold and that the penalty parameter $\rho$ is chosen sufficiently large.
Then the sequence $\{\Xi^{(k)}\}$ generated by Algorithm~\ref{alg:admm_topology} has the following properties:
\begin{enumerate}
    \item The consensus residual converges to zero,
    \[
        d^{(k)} \to 0 \quad \text{as } k\to\infty,
    \]
    so that the coupling constraint $z - r + s = 0$ is satisfied at any limit point.
    \item The sequence has square-summable increments, in the sense that
    \begin{align}
        \sum_{k=0}^{\infty} \Big(
       & \big\|(z^{(k+1)},f^{(k+1)}) - (z^{(k)},f^{(k)})\big\|_2^2
        \nonumber\\&+ \big\|s^{(k+1)} - s^{(k)}\big\|_2^2
        \Big) < +\infty. 
    \end{align}
    \item Every limit point of $\{\Xi^{(k)}\}$ is a stationary point of the augmented Lagrangian $\mathcal{L}$ under the convex constraints on $(z,f)$ and the binary constraints on $r$.
\end{enumerate}
In particular, if the sequence $\{\Xi^{(k)}\}$ has a unique limit point, then the entire sequence converges to this stationary point.
\end{theorem}

\section{Quantum Optimization of the Binary Block}
\label{sec:quantum_block}

In this section, we describe how the binary Block~2 of the three-block ADMM is mapped to a QUBO problem and then encoded as a quantum Hamiltonian. The ground state of this Hamiltonian corresponds to an optimal or near-optimal binary edge configuration for the communication topology. We then explain how the QITE algorithm is used as a quantum solver for this Hamiltonian.

\subsection{Binary ADMM Block as QUBO Hamiltonian}
\label{subsec:qubo_block}

In the three-block ADMM scheme in Section~\ref{sec:three_block_admm}, the binary block updates the proxy vector $r \in \{0,1\}^m$ by minimizing the augmented Lagrangian with respect to $r$, while all other variables are fixed. The Block~2 subproblem in \eqref{eq:block2_r_problem} can be written in the compact form
\begin{equation}
  \min_{r \in \{0,1\}^m}
   \Phi(r)
  := \lambda^\top (z - r + s)
   + \frac{\rho}{2}\|z - r + s\|_2^2
   + \mu \Big(\sum_{i=1}^m r_i\Big)^2,
  \label{eq:block2_cost}
\end{equation}
where, for simplicity of notation, we dropped the iteration superscripts and we used $\mu \ge 0$ to denote a small cardinality penalty.

The first two terms in \eqref{eq:block2_cost} come directly from the augmented Lagrangian in \eqref{eq:aug_lagrangian_three_block}. If we only keep these two terms, the objective is separable in the components of $r$. Indeed, the quadratic term in $\|z - r + s\|_2^2$ produces diagonal contributions in the Hessian with respect to $r$, and the linear term $\lambda^\top(z - r + s)$ only shifts each variable individually. In this ideal case, every $r_i$ could be updated independently by comparing two scalar costs, since $r_i \in \{0,1\}$.

However, in practice, we observed that a completely separable binary block can lead to slower convergence and oscillations in the ADMM iterations, because each edge variable reacts locally to the current $z$ and $s$ values, without any direct coupling with the other edges. Moreover, the separable structure does not encode any preference for the total number of active links. For these reasons, we add the heuristic cardinality regularization
\[
  J_{\text{card}}(r)
  := \mu \Big(\sum_{i=1}^m r_i\Big)^2,
\]
which appears as the last term in \eqref{eq:block2_cost}. This penalty is small compared to the main terms, but it has two useful effects. First, it discourages both extremely dense and extremely sparse topologies, because the cost grows quadratically with the number of active edges. Second, it introduces a controlled coupling between all binary variables. This coupling produces a non-diagonal quadratic form in $r$ and tends to stabilize the ADMM dynamics by correlating the updates of different edges.

To make the quantum mapping explicit, we rewrite \eqref{eq:block2_cost} in the standard QUBO form. If we expand the squared norm and the cardinality term, we can always write
\begin{equation}
  \Phi(r) = r^\top Q r + q^\top r,
  \quad r \in \{0,1\}^m,
  \label{eq:qubo_standard}
\end{equation}
where $Q \in \mathbb{R}^{m \times m}$ is a symmetric matrix collecting all quadratic coefficients and $q \in \mathbb{R}^m$ is a linear coefficient vector. In particular, the diagonal of $Q$ and the entries of $q$ come from the expansion of $\lambda^\top(z - r + s)$ and $\|z - r + s\|_2^2$, while the off-diagonal entries of $Q$ are created by the cardinality term $J_{\text{card}}(r)$. In the absence of the cardinality penalty, $Q$ would be diagonal and the problem would be separable.

The QUBO in \eqref{eq:qubo_standard} is a purely combinatorial problem in the binary vector $r$. To use a quantum algorithm, it is convenient to transform it into a Hamiltonian acting on $m$ qubits. This is done by introducing Pauli-$Z$ spin operators and by encoding each binary decision into a spin variable. Following a standard construction, we associate to each QUBO variable $r_i \in \{0,1\}$ a Pauli-$Z$ operator $Z_i$ with eigenvalues $\pm 1$ and we use the linear mapping
\begin{equation}
  r_i = \frac{1 - Z_i}{2},
  \quad i = 1,\dots,m.
  \label{HamiltonianT}
\end{equation}

If we substitute \eqref{HamiltonianT} into \eqref{eq:qubo_standard}, and we expand all products $r_i r_j$ in terms of $Z_i$ and $Z_j$, we obtain an operator of the form
\begin{equation}
  H = \sum_{i=1}^m h_i Z_i
      + \sum_{1 \le i < j \le m} J_{ij} Z_i Z_j
      + \text{const},
  \label{eq:ising_form}
\end{equation}
where $h_i$ and $J_{ij}$ are real coefficients that can be computed by simple algebraic manipulations of $Q$ and $q$, and the constant shift does not change the ground state. This operator $H$ is a Hamiltonian acting on $m$ qubits, written as a linear combination of tensor products of Pauli-$Z$ matrices. 

From a quantum computing point of view, the Hamiltonian $H$ in \eqref{eq:ising_form} is the encoding of the topology design binary block into the language of quantum mechanics. Each bitstring $r \in \{0,1\}^m$ corresponds to a basis state of the $m$-qubit system, and the QUBO cost becomes the energy of that state. The optimal edge configuration is obtained by finding the ground state of $H$, that is, the eigenstate with the minimum eigenvalue. 

\subsection{QITE Solver}
\label{subsec:qite}

After we map the binary ADMM block to the Hamiltonian in \eqref{eq:ising_form}, the problem becomes the following: we want to find the ground quantum state of $H$, that is, the quantum state with the smallest energy. Each bitstring corresponds to one quantum basis state, and the energy of that quantum state is equal to the QUBO cost. If we can prepare the quantum ground state of $H$, we can read out a good binary solution by measuring it. Note that each time we talk about a state in this section is a quantum state and not a state of MAS. We can think of the energy of $H$ as a landscape with many hills and valleys. The ground state is the bottom of the deepest valley. 

The QITE algorithm is a way to slide down this landscape. In quantum mechanics, time evolution is usually defined with the Schrödinger equation. If we replace the real-time variable $t$ with $-{\rm i}\tau$, where ${\rm i}$ is the imaginary unit, we obtain an imaginary-time equation. This imaginary-time evolution damps the high-energy components of the state while preserving the low-energy components. After some imaginary time, the state tends to align with the ground state of $H$.

Implementing this ideal imaginary-time evolution directly is difficult on current noisy quantum hardware, because it would require non-unitary operations. Instead, QITE uses an approximate, variational approach. The main idea is to restrict the search for the ground state to a family of states generated by a fixed quantum circuit with tunable parameters. We write
\[
  |\psi(\theta)\rangle = U(\theta)\,|0\rangle,
\]
where $U(\theta)$ is called an ansatz and $\theta$ is a vector of real parameters. In our setting, $U(\theta)$ is a hardware-efficient circuit acting on the $m$ qubits associated with the edge variables. It contains several layers of single-qubit rotations and entangling gates between nearby qubits. The structure of $U(\theta)$ is fixed in advance, and only the rotation angles in these gates are adjusted.

\begin{figure}[t]
  \centering
  \includegraphics[width=0.85\linewidth]{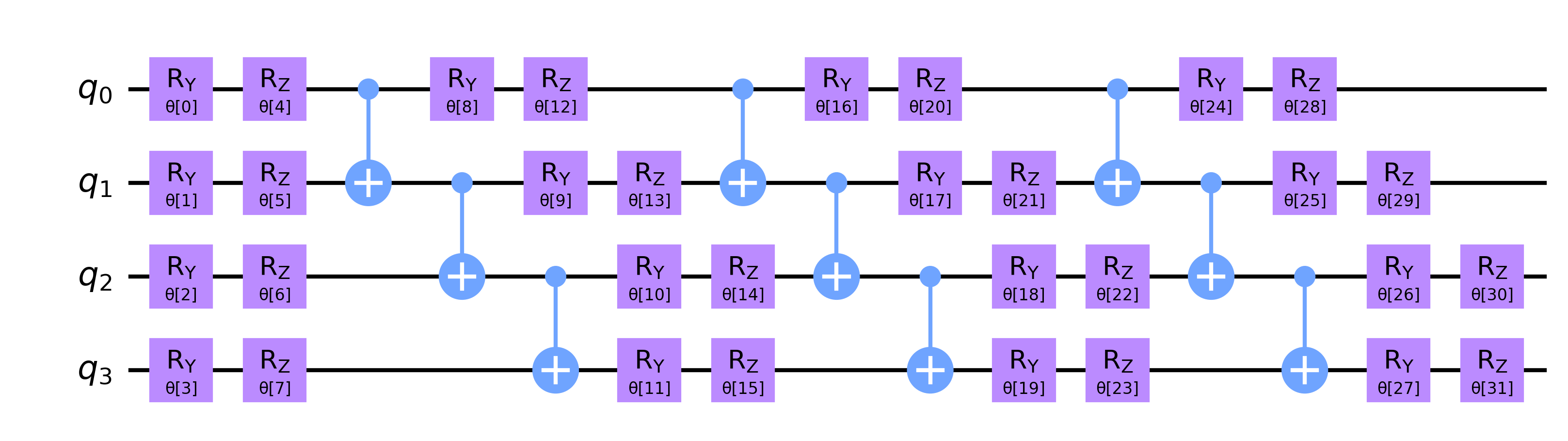}
  \caption{Example of an ansatz circuit $U(\theta)$ for QITE}
  \label{fig:qite_ansatz}
\end{figure}

QITE then interprets the imaginary-time evolution as a differential equation, not on the full state space, but on the parameter vector $\theta$. Conceptually, we want $\theta(\tau)$ to evolve in such a way that the ansatz state $|\psi(\theta(\tau))\rangle$ follows the ideal imaginary-time evolution of $H$ as closely as possible. This is why QITE is often described as a quantum differential equation solver: the imaginary unit in the original Schrödinger equation leads to an imaginary-time equation, and QITE approximates its solution by updating the circuit parameters rather than explicitly evolving the full wavefunction.

In practice, the algorithm proceeds in small imaginary-time steps. At each step, the following operations are performed. First, the current parameters $\theta$ are fixed and the ansatz circuit $U(\theta)$ is executed on the quantum device to prepare the current state $|\psi(\theta)\rangle$. Second, expectation values of certain observables that involve the Hamiltonian $H$ are measured, and the derivatives of the state with respect to the parameters are calculated. These measurements provide the numerical coefficients of a small linear system that describes how $\theta$ should change to approximate one step of imaginary-time evolution. Third, the small linear system is solved, yielding an update $\Delta\theta$. Finally, the parameters are moved in that direction, for example, \( \theta \leftarrow \theta + \Delta\theta.\) This procedure is then repeated for several steps. Each step is designed to decrease the energy expectation $\langle \psi(\theta) | H | \psi(\theta) \rangle$, so the parameters gradually move toward a region where the ansatz state has low energy.
\begin{figure}[t]
  \centering
  \begin{tikzpicture}[
    scale=0.75,
    transform shape,
    font=\scriptsize,
    node distance=5mm,
    >=Stealth,
    every node/.style={align=center},
    dataBox/.style={
      draw,
      rounded corners,
      thick,
      fill=gray!5,
      draw=gray!60,
      inner sep=2pt,
      minimum width=28mm,
      minimum height=5mm
    },
    quantumBox/.style={
      draw,
      rounded corners,
      thick,
      fill=blue!5,
      draw=blue!60,
      inner sep=2pt,
      minimum width=28mm,
      minimum height=5mm
    },
    classicalBox/.style={
      draw,
      rounded corners,
      thick,
      fill=orange!5,
      draw=orange!70!black,
      inner sep=2pt,
      minimum width=28mm,
      minimum height=5mm
    },
    line/.style={-Stealth, thick}
  ]

    % Nodes (vertical flow)
    \node[dataBox]      (inputH) {QUBO data $(Q,q)$\\and Hamiltonian $H$};
    \node[dataBox,      below=of inputH] (init)
        {Initial parameters\\$\theta^{(0)}$};
    \node[quantumBox,   below=of init]   (prepare)
        {Prepare ansatz state\\$|\psi(\theta)\rangle = U(\theta)|0\rangle$\\on quantum device};
    \node[quantumBox,   below=of prepare] (measure)
        {Measure expectation\\values w.r.t.\ $H$};
    \node[classicalBox, below=of measure] (linear)
        {Build and solve small\\linear system for $\Delta\theta$};
    \node[classicalBox, below=of linear]  (update)
        {Update parameters\\$\theta \leftarrow \theta + \Delta\theta$};
    \node[dataBox,      below=of update, yshift=-1mm] (output)
        {Approximate ground-state\\parameters $\theta^\star$ and\\low-energy bitstrings};

    % Loop label on the right
    \node[anchor=west, xshift=7mm] (looplabel) at ($(measure.east)!0.5!(update.east)$)
        {Imaginary-time\\steps};

    % Main downward arrows
    \draw[line] (inputH) -- (init);
    \draw[line] (init)   -- (prepare);
    \draw[line] (prepare) -- (measure);
    \draw[line] (measure) -- (linear);
    \draw[line] (linear)  -- (update);
    \draw[line] (update)  -- (output);

    % Feedback loop from update back to prepare
    \draw[line]
      ([xshift=4mm]update.east) to[out=0,in=0,looseness=1.2]
      ([xshift=4mm]prepare.east);

    % Dashed box highlighting the iterative QITE loop
    \draw[dashed, rounded corners]
      ($(prepare.north west)+(-2mm,2mm)$) rectangle
      ($(update.south east)+(2mm,-2mm)$);

  \end{tikzpicture}
  \caption{Schematic view of the QITE procedure}
  \label{fig:qite_workflow}
\end{figure}
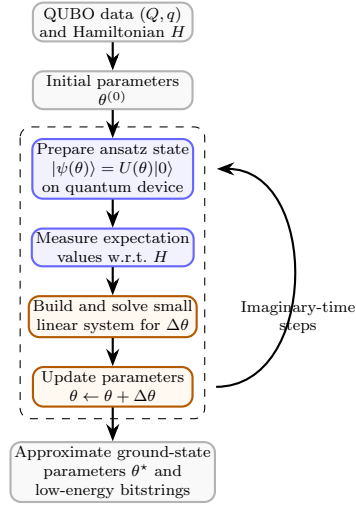

After a finite number of imaginary-time steps, we obtain a parameter vector $\theta^\star$ for which the ansatz state $|\psi(\theta^\star)\rangle$ has low energy with respect to $H$. This state is an approximation of the ground quantum state of the QUBO Hamiltonian. To extract a binary topology from it, we measure $|\psi(\theta^\star)\rangle$ in the computational basis many times. Each measurement produces a bitstring $r \in \{0,1\}^m$, which encodes one candidate set of active edges. The probabilities of these bitstrings depend on the amplitudes of the quantum state. In our implementation, we collect the bitstrings with the highest empirical probabilities, evaluate their exact QUBO cost $z^\top Q z + q^\top z$ on a classical computer, and select the bitstring with the lowest cost as the solution of the binary block for that ADMM iteration.

In the overall three-block ADMM algorithm, this QITE-based solver is used as a module that proposes a binary proxy vector $r$. At each ADMM iteration, the current QUBO is converted to a Hamiltonian, QITE is run for a limited number of imaginary-time steps, and a small set of candidate bitstrings is sampled and evaluated. If the new candidate improves the QUBO objective, it is accepted and used as the updated $r$; otherwise, the previous binary solution can be kept. In this way, the quantum routine explores the combinatorial space of communication topologies, while the ADMM framework maintains a clear structure and convergence properties for the full topology design algorithm.

The overall closed-loop dynamic procedure that combines ADMM-based topology optimization, QITE-based binary updates, and consensus dynamics is summarized in Algorithm~\ref{alg:closed_loop_q_admm}.

\begin{algorithm}[!t]
  \small
  \caption{Dynamic topology design}
  \label{alg:closed_loop_q_admm}
  \begin{algorithmic}[1]
    \State \textbf{Inputs:} initial states $\boldsymbol{x}(0)$ (and $\boldsymbol{v}(0)$), update times $\{t_\ell\}$, degree bound $\gamma$, ADMM parameters $\rho,\beta,\mu$, QITE settings (ansatz, steps, shots).
    \State Choose initial connected graph $\mathcal{G}_0$ and Laplacian $L_0$ satisfying Assumption~\ref{ass:connectivity}.
    \State Set $\ell \gets 0$.
    \While{closed-loop operation is active}
      \State \textbf{Consensus evolution on $[t_\ell,t_{\ell+1})$:}
      \State Simulate consensus with $L_\ell$:
      \State \hspace{0.4cm}First-order: \eqref{eq:csys1_tv}.
      \State \hspace{0.4cm}Second-order: \eqref{eq:csys2_tv_x}--\eqref{eq:csys2_tv}.
      \State At $t_{\ell+1}$, measure agent positions (and velocities).
      \State Update edge weights $w_{ij}$ via \eqref{eq:edge_weight}.
      \State Form MIQP \eqref{eq:MIQP_topology} with costs \eqref{eq:topology_linear_cost}, \eqref{eq:degree_quadratic_cost} and constraints \eqref{eq:MIQP_deg}--\eqref{eq:MIQP_flow_root}.
      \State \textbf{Initialize ADMM:}
      \State Choose $(z^{(0)},f^{(0)},r^{(0)},s^{(0)},\lambda^{(0)})$, set $k \gets 0$.
      \While{ADMM stopping criterion not met}
        \State $k \gets k + 1$.
        \State \textbf{Block 1 (classical convex subproblem):}
        \State Given $r^{(k-1)},s^{(k-1)},\lambda^{(k-1)}$, solve \eqref{eq:block1_full}
        \State for $(z^{(k)},f^{(k)})$ over $\mathcal{C}$.
        \State \textbf{Block 2 (quantum binary subproblem):}
        \State Build QUBO cost $\Phi(r)$ in \eqref{eq:block2_cost}.
        \State Map QUBO to Ising $H^{(k)}$ via \eqref{HamiltonianT}, \eqref{eq:ising_form}.
        \State Run QITE for $H^{(k)}$:
        \State \hspace{0.4cm}Initialize ansatz $U(\theta)$ (Fig.~\ref{fig:qite_ansatz}).
        \For{each imaginary-time step}
          \State Prepare $|\psi(\theta)\rangle = U(\theta)|0\rangle$.
          \State Measure expectations required by QITE.
          \State Build and solve the small linear system for $\Delta\theta$.
          \State Update $\theta \gets \theta + \Delta\theta$.
        \EndFor
        \State Sample bitstrings $r$ from $|\psi(\theta)\rangle$.
        \State Evaluate $\Phi(r)$ and set $r^{(k)}$ to the best sample.
        \State \textbf{Block 3 (auxiliary update):}
        \State Update $s^{(k)}$ using \eqref{eq:block3_closed_form}.
        \State \textbf{Dual update:}
        \State Compute $d^{(k)}$ from \eqref{residual} and update $\lambda^{(k)}$ via \eqref{eq:dual_update_three_block}.
        \State Check convergence as in Theorem~\ref{thm:admm_convergence}.
      \EndWhile
      \State \textbf{Laplacian update:}
      \State Construct $\mathcal{E}_{\ell+1}$ and $L_{\ell+1}$ from the ADMM output
      \State (e.g., using $r^{(k)}$ or thresholded $z^{(k)}$),
      \State ensuring \eqref{eq:MIQP_deg}--\eqref{eq:MIQP_flow_root} hold.
      \State Set $\ell \gets \ell + 1$.
    \EndWhile
  \end{algorithmic}
\end{algorithm}

\section{Numerical Results}
\label{sec:numerical}

We present numerical examples that demonstrate the proposed dynamic quantum topology design method. All experiments are generated by a single Python script that follows Algorithm~\ref{alg:closed_loop_q_admm}. The simulations are carried out on a classical computer, and the quantum routines are executed on a Qiskit simulator rather than on real quantum hardware. We use Qiskit version 0.46.3 and qiskit-algorithms version 0.3.1. The QITE solver is implemented using VarQITE from qiskit-algorithms, combined with the ImaginaryMcLachlanPrinciple and an EfficientSU2 ansatz with linear entanglement. The complete code is available at \href{https://github.com/LSU-RAISE-LAB/DQOCT}{GitHub}. 

\subsection{Simulation Setup and Implementation}
\label{subsec:simulation_setup}

The code starts by fixing the number of agents, the order of the consensus dynamics, and the initial conditions. The underlying candidate graph is the complete undirected graph on \(n\) nodes, and from this edge list we build a simple incidence matrix used to compute degrees and to assemble the Laplacian from a given edge indicator vector. In each topology update, the script calls the ADMM routine that implements the three-block splitting described in Section~\ref{sec:three_block_admm}. 

Block 2 constructs the QUBO in the binary proxy variable. From the resulting quadratic and linear coefficients, we build an Ising Hamiltonian. When the brute-force option is active, the code evaluates the exact QUBO cost over all bitstrings and selects the best one, which is useful for debugging and for small graphs. When QITE is used, VarQITE evolves an EfficientSU2 ansatz under imaginary time with respect to this Hamiltonian. 

The ADMM hyperparameters are fixed at the beginning of each run. The penalty parameter \(\rho\) is kept the same in all experiments, while the penalty on the auxiliary variable, denoted by \(\beta\), is chosen as a simple function of the number of agents. In the implementation \(\beta\) takes moderate values (around \(10^2\)) for small networks and grows gradually for larger networks according to a geometric rule. These values were tuned by trial and error to achieve stable residual histories and reasonable convergence speed, and they can be adjusted by the user for different problem sizes.

After convergence of the ADMM loop at a given snapshot, the relaxed edge vector is thresholded at 0.5 to produce a binary topology. From this binary vector, we construct the Laplacian and compute its eigenvalues to verify that the resulting graph is connected. The Laplacian is then used in the consensus dynamics. The final state at the end of the interval becomes the initial condition for the next topology update. At each step, the code computes a consensus error. The online loop stops when this error falls below a chosen tolerance or when a maximum simulation time is reached. In some experiments, we also allow random perturbations of the positions when the error remains large for a long time, so that the algorithm must reconfigure the topology.

\subsection{Numerical Examples}
\label{subsec:numerical_examples}

In all examples, the candidate communication graph is the complete undirected graph on \(n\) nodes, and the degree bound is fixed at \(\gamma = 2\). The degree penalty coefficients are chosen as \(  \kappa_i = 0.1, \; i = 1,\dots,n,\) so that highly unbalanced degree profiles are mildly penalized. The ADMM hyperparameters are \(\rho = 20, \qquad \mu_{\text{card}} = 0.1,\) while the auxiliary penalty \(\beta\) is chosen as a function of the network size \(n = N_{\text{AGENTS}}\):
\[
  \beta =
  \begin{cases}
    200, & n = 5,\\
    400, & n = 6,\\
    600, & n = 7.
  \end{cases}
\]
The maximum number of ADMM iterations at each topology update is set to \( k_{\max} = 500\) ADMM tolerance \(10^{-3}.\) For the second-order consensus dynamics we fix the control gains \( \alpha = 3.0, \qquad \beta = 3.0,\) as in \eqref{eq:csys2_tv}. The initial positions are sampled uniformly at random in the interval \([-5,5]\), \(  x_i(0) \sim \mathcal{U}(-5,5),\) and the initial velocities in the second-order cases are set to zero, \( v_i(0) =0 .\)

The initial topology is a connected graph that satisfies the degree bound and Assumption~\ref{ass:connectivity}, and it is held constant until \(t = 5~\mathrm{s}\). After this first update time, the topology is recomputed every \(\Delta t = 0.5~\mathrm{s}\) using Algorithm~\ref{alg:closed_loop_q_admm}. The overall closed-loop simulation is capped at \(T_{\max} = 10~\mathrm{s}.\)

To monitor convergence we use a standard consensus error metric. The closed-loop run is terminated as soon as the position disagreement drops below the tolerance
\[
  e(t) \le \texttt{CONS\_TOL} = 10^{-3}
\]
or when \(t\) reaches \(T_{\max}\). The quantum binary block uses VarQITE with imaginary-time horizon \(\tau = 1.5\), discretized in \(30\) steps, an EfficientSU2 ansatz with one repetition and linear entanglement.

\subsubsection*{Example 1: Five Agents with First-Order Dynamics}

In the first example we consider a network of \(n=5\) agents with first-order consensus dynamics \eqref{eq:csys1_tv}. 

Figure~\ref{fig:ex1_states} shows the evolution of the agent positions. Starting from dispersed initial conditions, the trajectories converge to a common value close to the initial average, in agreement with Theorem~\ref{thm:consensus}.

\begin{figure}[!t]
    \centering
    \includegraphics[width=0.9\linewidth]{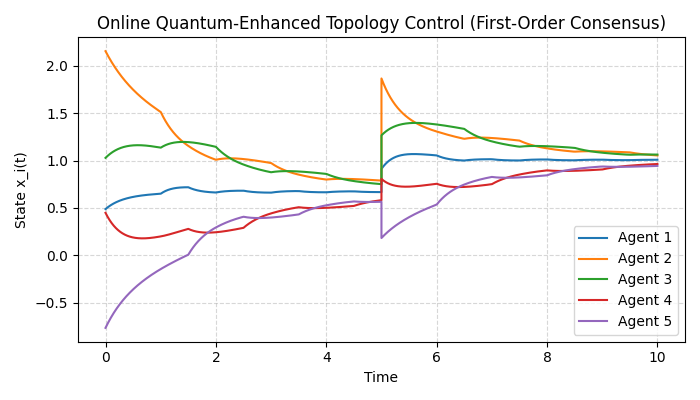}
    \caption{\footnotesize Example~1 (first-order consensus with \(n=5\) agents): evolution of the agent positions}
    \label{fig:ex1_states}
\end{figure}

Figure~\ref{fig:ex1_graph} depicts the optimized graphs obtained during the simulation. All node degrees satisfy \(\deg(i)\le \gamma = 2\), and the resulting topology is sparse and close to a path-like structure, which is consistent with the degree constraints and with the quadratic degree penalty. 

\begin{figure}[!t]
    \centering
    \includegraphics[width=0.9\linewidth]{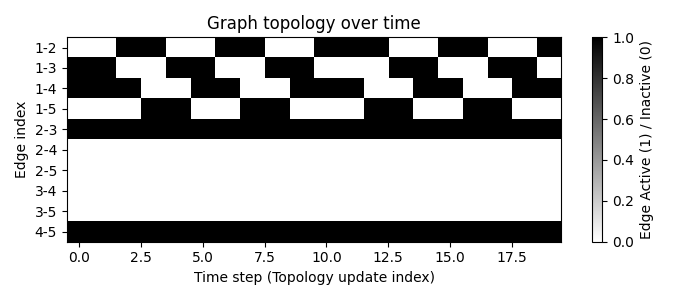}
    \caption{\footnotesize optimized communication graphs}
    \label{fig:ex1_graph}
\end{figure}

The corresponding consensus error \eqref{eq:csys1_tv} is plotted in Fig.~\ref{fig:ex1_error}. 

\begin{figure}[!t]
    \centering
    \includegraphics[width=0.9\linewidth]{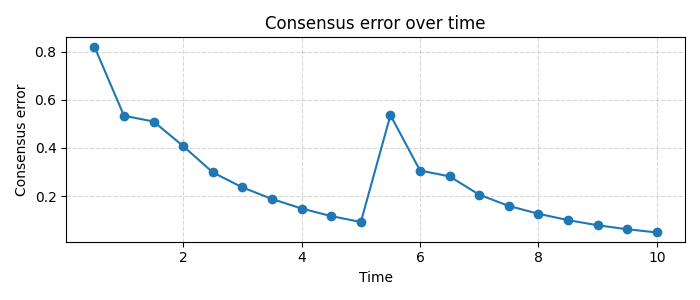}
    \caption{\footnotesize consensus error versus time}
    \label{fig:ex1_error}
\end{figure}

\subsubsection*{Example 2: Six Agents with Second-Order Dynamics}

In the second example, we consider a network of \(n=6\) agents with second-order consensus dynamics \eqref{eq:csys2_tv_x}--\eqref{eq:csys2_tv}. 

Figure~\ref{fig:ex2_states_pos} shows the evolution of the agent positions. The trajectories converge to a common value, consistent with second-order consensus.

\begin{figure}[!t]
    \centering
    \includegraphics[width=0.9\linewidth]{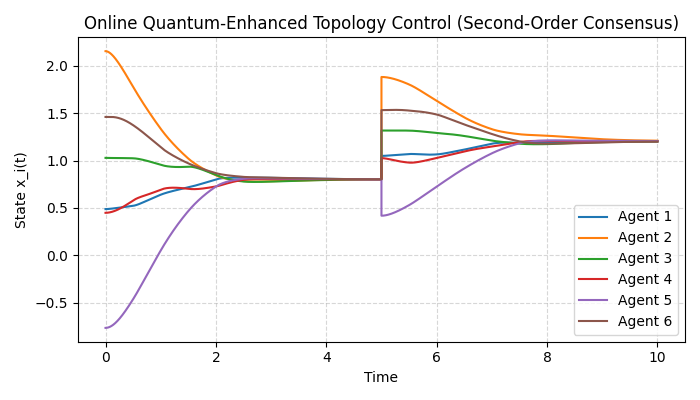}
    \caption{\footnotesize Example~2 (second-order consensus with \(n=6\) agents): evolution of the agent positions}
    \label{fig:ex2_states_pos}
\end{figure}

Figure~\ref{fig:ex2_states_vel} shows the corresponding velocity trajectories, which converge to zero for all agents.

\begin{figure}[!t]
    \centering
    \includegraphics[width=0.9\linewidth]{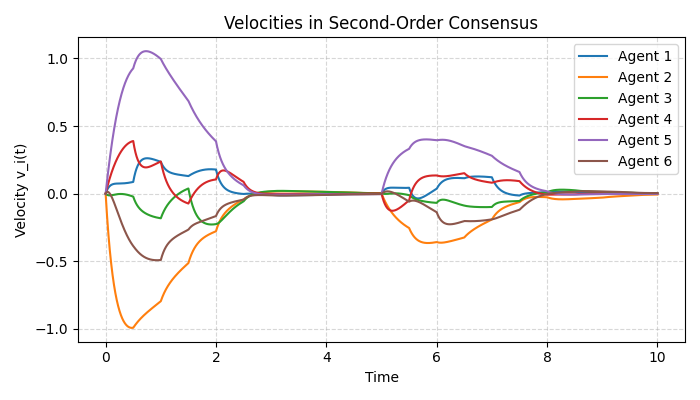}
    \caption{\footnotesize Example~2: evolution of the agent velocities}
    \label{fig:ex2_states_vel}
\end{figure}

Figure~\ref{fig:ex2_graph} depicts the optimized communication graphs obtained during the simulation. All node degrees satisfy \(\deg(i)\le \gamma = 2\), and the resulting topologies remain sparse and connected.

\begin{figure}[!t]
    \centering
    \includegraphics[width=0.9\linewidth]{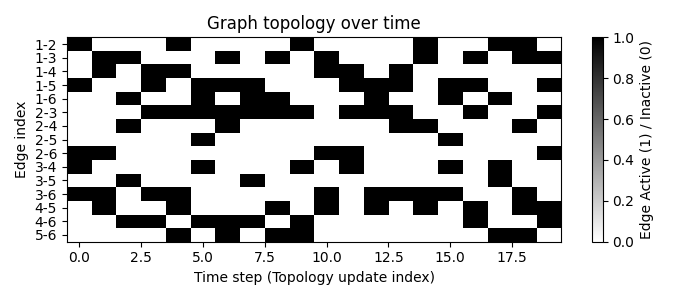}
    \caption{\footnotesize Example~2: optimized communication graphs}
    \label{fig:ex2_graph}
\end{figure}

The corresponding position and velocity consensus errors are plotted in Fig.~\ref{fig:ex2_error}. 

\begin{figure}[!t]
    \centering
    \includegraphics[width=0.9\linewidth]{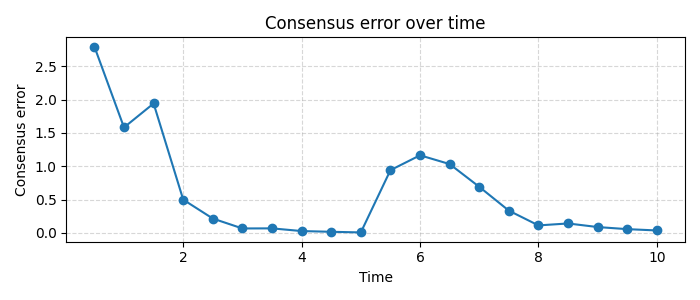}
    \caption{\footnotesize Example~2: position and velocity consensus errors versus time}
    \label{fig:ex2_error}
\end{figure}

\subsubsection*{Example 3: Seven Agents with Second-Order Dynamics}

In the third example, we consider a network of \(n=7\) agents with second-order consensus dynamics \eqref{eq:csys2_tv_x}--\eqref{eq:csys2_tv}. 

Figure~\ref{fig:ex3_states_pos} shows the evolution of the agent positions. Despite the larger network, the trajectories converge to a common value.

\begin{figure}[!t]
    \centering
    \includegraphics[width=0.9\linewidth]{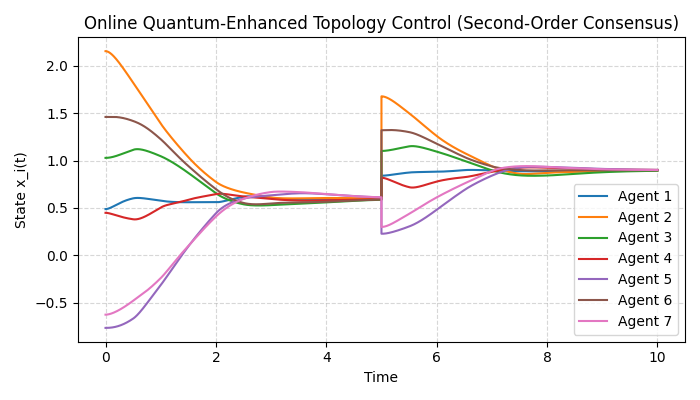}
    \caption{\footnotesize Example~3 (second-order consensus with \(n=7\) agents): evolution of the agent positions}
    \label{fig:ex3_states_pos}
\end{figure}

Figure~\ref{fig:ex3_states_vel} shows the velocity trajectories, which converge to zero for all agents.

\begin{figure}[!t]
    \centering
    \includegraphics[width=0.9\linewidth]{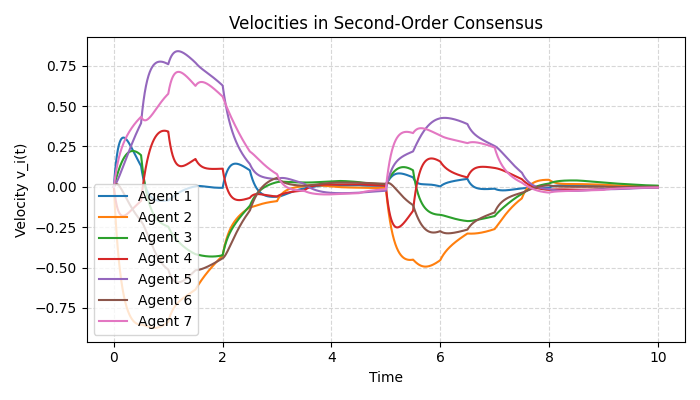}
    \caption{\footnotesize Example~3: evolution of the agent velocities}
    \label{fig:ex3_states_vel}
\end{figure}

Figure~\ref{fig:ex3_graph} depicts the optimized communication graphs obtained during the simulation. The topologies are sparse, degree-bounded, and remain connected throughout the run.

\begin{figure}[!t]
    \centering
    \includegraphics[width=0.9\linewidth]{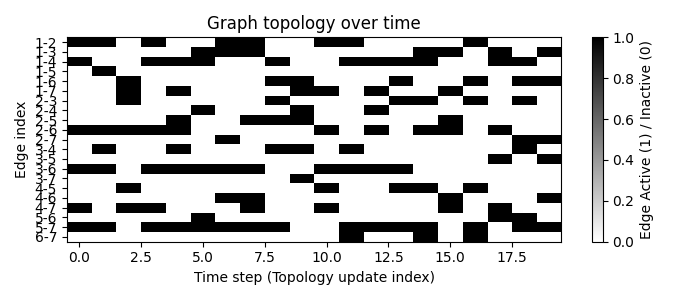}
    \caption{\footnotesize Example~3: optimized communication graphs}
    \label{fig:ex3_graph}
\end{figure}

The corresponding position and velocity consensus errors are plotted in Fig.~\ref{fig:ex3_error}. 

\begin{figure}[!t]
    \centering
    \includegraphics[width=0.9\linewidth]{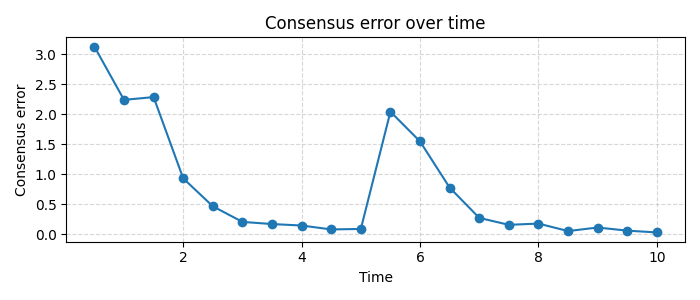}
    \caption{\footnotesize Example~3: position and velocity consensus errors versus time}
    \label{fig:ex3_error}
\end{figure}

\section{Conclusion}
\label{sec:conclusion}

We proposed a dynamic quantum framework for online communication topology design in consensus-based linear multi-agent systems, formulating the problem as an MIQP with explicit communication, distance, and degree costs, and exact flow-based connectivity constraints. A three-block ADMM splitting was introduced that keeps all graph constraints in a convex classical block, isolates a pure binary unconstrained block, and couples them via a simple auxiliary update, so that only the binary core is delegated to a quantum solver. The binary block is mapped to a QUBO Hamiltonian and approximately solved by QITE, which embeds a fully quantum ground-state routine within a structured classical decomposition while preserving first- and second-order consensus guarantees under mild connectivity and dwell-time assumptions. Numerical simulations with five to seven agents show that the method produces sparse, degree-bounded, connected topologies that achieve consensus and communication costs, illustrating a concrete pathway for integrating near-term quantum optimizers into closed-loop distributed control.

\bibliographystyle{plain}
\bibliography{example}  

@inproceedings{hasanzadeh2024dynamic,
  title={Dynamic average consensus as distributed {PDE}-based control for multi-agent systems},
  author={Hasanzadeh, Milad and Tang, Shu-Xia},
  booktitle={2024 European Control Conference (ECC)},
  pages={828--833},
  year={2024},
  organization={IEEE}
}

@article{su2015distributed,
  title={Distributed consensus control of multi-agent systems with higher order agent dynamics and dynamically changing directed interaction topologies},
  author={Su, Shize and Lin, Zongli},
  journal={IEEE Transactions on Automatic Control},
  volume={61},
  number={2},
  pages={515--519},
  year={2015},
  publisher={IEEE}
}

@article{olfati2007consensus,
  title={Consensus and cooperation in networked multi-agent systems},
  author={Olfati-Saber, Reza and Fax, J Alex and Murray, Richard M},
  journal={Proceedings of the IEEE},
  volume={95},
  number={1},
  pages={215--233},
  year={2007},
  publisher={IEEE}
}

@inproceedings{rafiee2010optimal,
  title={Optimal network topology design in multi-agent systems for efficient average consensus},
  author={Rafiee, Mohammad and Bayen, Alexandre M},
  booktitle={49th IEEE Conference on Decision and Control (CDC)},
  pages={3877--3883},
  year={2010},
  organization={IEEE}
}

@article{Preskill2018quantumcomputingin,
  doi = {10.22331/q-2018-08-06-79},
  url = {https://doi.org/10.22331/q-2018-08-06-79},
  title = {Quantum {C}omputing in the {NISQ} era and beyond},
  author = {Preskill, John},
  journal = {{Quantum}},
  issn = {2521-327X},
  publisher = {{Verein zur F{\"{o}}rderung des Open Access Publizierens in den Quantenwissenschaften}},
  volume = {2},
  pages = {79},
  month = aug,
  year = {2018}
}

@article{kargarian2016toward,
  title={Toward distributed/decentralized DC optimal power flow implementation in future electric power systems},
  author={Kargarian, Amin and Mohammadi, Javad and Guo, Junyao and Chakrabarti, Sambuddha and Barati, Masoud and Hug, Gabriela and Kar, Soummya and Baldick, Ross},
  journal={IEEE Transactions on Smart Grid},
  volume={9},
  number={4},
  pages={2574--2594},
  year={2016},
  publisher={IEEE}
}

@article{farhi2014quantum,
  title={A quantum approximate optimization algorithm},
  author={Farhi, Edward and Goldstone, Jeffrey and Gutmann, Sam},
  journal={arXiv preprint arXiv:1411.4028},
  year={2014}
}

@article{peruzzo2014variational,
  title={A variational eigenvalue solver on a photonic quantum processor},
  author={Peruzzo, Alberto and McClean, Jarrod and Shadbolt, Peter and Yung, Man-Hong and Zhou, Xiao-Qi and Love, Peter J and Aspuru-Guzik, Al{\'a}n and O’brien, Jeremy L},
  journal={Nature communications},
  volume={5},
  number={1},
  pages={4213},
  year={2014},
  publisher={Nature Publishing Group UK London}
}

@article{yu2009second,
  title={Second-order consensus for multiagent systems with directed topologies and nonlinear dynamics},
  author={Yu, Wenwu and Chen, Guanrong and Cao, Ming and Kurths, J{\"u}rgen},
  journal={IEEE Transactions on Systems, Man, and Cybernetics, Part B (Cybernetics)},
  volume={40},
  number={3},
  pages={881--891},
  year={2009},
  publisher={IEEE}
}

@article{mesbahi2010graph,
  title={Graph theoretic methods in multiagent networks},
  author={Mesbahi, Mehran and Egerstedt, Magnus},
  year={2010},
  publisher={Princeton University Press}
}

@article{jenefa2024enhancing,
  title={Enhancing distributed agent environments with quantum multi-agent systems and protocols},
  author={Jenefa, A and Vidhya, K and Taurshia, Antony and Naveen, V Edward and Kuriakose, Bessy M and Vijula, V},
  journal={Multiagent and Grid Systems},
  volume={20},
  number={2},
  pages={109--127},
  year={2024},
  publisher={SAGE Publications Sage UK: London, England}
}

@article{zhao2023quantum,
  title={Quantum Multi-Agent Reinforcement Learning as an Emerging AI Technology: A Survey and Future Directions},
  author={Zhao, Jun and Yu, Wenhan},
  journal={Authorea Preprints},
  year={2023},
  publisher={Authorea}
}

@article{acha2025application,
  title={Application of quantum telecommunication in multi-agent system},
  author={Acha, Stefalo and Yi, Sun},
  journal={Discover Robotics},
  volume={1},
  number={1},
  pages={3},
  year={2025},
  publisher={Springer}
}

@book{nielsen2010quantum,
  title={Quantum computation and quantum information},
  author={Nielsen, Michael A and Chuang, Isaac L},
  year={2010},
  publisher={Cambridge university press}
}

@article{shor1999polynomial,
  title={Polynomial-time algorithms for prime factorization and discrete logarithms on a quantum computer},
  author={Shor, Peter W},
  journal={SIAM review},
  volume={41},
  number={2},
  pages={303--332},
  year={1999},
  publisher={SIAM}
}

@inproceedings{grover1996fast,
  title={A fast quantum mechanical algorithm for database search},
  author={Grover, Lov K},
  booktitle={Proceedings of the twenty-eighth annual ACM symposium on Theory of computing},
  pages={212--219},
  year={1996}
}

@article{motta2020determining,
  title={Determining eigenstates and thermal states on a quantum computer using quantum imaginary time evolution},
  author={Motta, Mario and Sun, Chong and Tan, Adrian TK and O’Rourke, Matthew J and Ye, Erika and Minnich, Austin J and Brandao, Fernando GSL and Chan, Garnet Kin-Lic},
  journal={Nature Physics},
  volume={16},
  number={2},
  pages={205--210},
  year={2020},
  publisher={Nature Publishing Group UK London}
}

@article{somisetty2024optimal,
  title={Optimal robust network design: Formulations and algorithms for maximizing algebraic connectivity},
  author={Somisetty, Neelkamal and Nagarajan, Harsha and Darbha, Swaroop},
  journal={IEEE Transactions on Control of Network Systems},
  year={2024},
  publisher={IEEE}
}

@inproceedings{goemans2006minimum,
  title={Minimum bounded degree spanning trees},
  author={Goemans, Michel X},
  booktitle={2006 47th Annual IEEE Symposium on Foundations of Computer Science (FOCS'06)},
  pages={273--282},
  year={2006},
  organization={IEEE}
}

@article{vielma2015mixed,
  title={Mixed integer linear programming formulation techniques},
  author={Vielma, Juan Pablo},
  journal={Siam Review},
  volume={57},
  number={1},
  pages={3--57},
  year={2015},
  publisher={SIAM}
}

@article{burer2009nonconvex,
  title={On nonconvex quadratic programming with box constraints},
  author={Burer, Samuel and Letchford, Adam N},
  journal={SIAM Journal on Optimization},
  volume={20},
  number={2},
  pages={1073--1089},
  year={2009},
  publisher={SIAM}
}

@article{de2012min,
  title={Min-degree constrained minimum spanning tree problem: complexity, properties, and formulations},
  author={de Almeida, Ana Maria and Martins, Pedro and de Souza, Maur{\'\i}cio C},
  journal={International Transactions in Operational Research},
  volume={19},
  number={3},
  pages={323--352},
  year={2012},
  publisher={Wiley Online Library}
}

@article{ravi2001approximation,
  title={Approximation algorithms for degree-constrained minimum-cost network-design problems},
  author={Ravi, Ramamoorthi and Marathe, Madhav V and Ravi, SS and Rosenkrantz, Daniel J and Hunt III, Harry B},
  journal={Algorithmica},
  volume={31},
  number={1},
  pages={58--78},
  year={2001},
  publisher={Springer}
}

@article{montanaro2016quantum,
  title={Quantum algorithms: an overview},
  author={Montanaro, Ashley},
  journal={npj Quantum Information},
  volume={2},
  number={1},
  pages={1--8},
  year={2016},
  publisher={Nature Publishing Group}
}

@article{hasanzadeh2024distributed,
  title={Distributed fixed-time rotating encirclement control of linear multi-agent systems with moving targets},
  author={Hasanzadeh, Milad and Baradarannia, Mahdi and Hashemzadeh, Farzad},
  journal={Journal of the Franklin Institute},
  volume={361},
  number={11},
  pages={106970},
  year={2024},
  publisher={Elsevier}
}

@article{ajagekar2019quantum,
  title={Quantum computing for energy systems optimization: Challenges and opportunities},
  author={Ajagekar, Akshay and You, Fengqi},
  journal={Energy},
  volume={179},
  pages={76--89},
  year={2019},
  publisher={Elsevier}
}

@article{hong2016convergence,
  title={Convergence analysis of alternating direction method of multipliers for a family of nonconvex problems},
  author={Hong, Mingyi and Luo, Zhi-Quan and Razaviyayn, Meisam},
  journal={SIAM Journal on Optimization},
  volume={26},
  number={1},
  pages={337--364},
  year={2016},
  publisher={SIAM}
}

@article{wang2019global,
  title={Global convergence of ADMM in nonconvex nonsmooth optimization},
  author={Wang, Yu and Yin, Wotao and Zeng, Jinshan},
  journal={Journal of Scientific Computing},
  volume={78},
  number={1},
  pages={29--63},
  year={2019},
  publisher={Springer}
}

@article{guo2017convergence,
  title={Convergence of ADMM for multi-block nonconvex separable optimization models},
  author={Guo, Ke and Han, Deren and Wang, David ZW and Wu, Tingting},
  journal={Frontiers of Mathematics in China},
  volume={12},
  number={5},
  pages={1139--1162},
  year={2017},
  publisher={Springer}
}

@article{yang2010decentralized,
  title={Decentralized estimation and control of graph connectivity for mobile sensor networks},
  author={Yang, Peng and Freeman, Randy A and Gordon, Geoffrey J and Lynch, Kevin M and Srinivasa, Siddhartha S and Sukthankar, Rahul},
  journal={Automatica},
  volume={46},
  number={2},
  pages={390--396},
  year={2010},
  publisher={Elsevier}
}

@article{aragues2014distributed,
  title={Distributed algebraic connectivity estimation for undirected graphs with upper and lower bounds},
  author={Aragues, Rosario and Shi, Guodong and Dimarogonas, Dimos V and Sag{\"u}{\'e}s, Carlos and Johansson, Karl Henrik and Mezouar, Youcef},
  journal={Automatica},
  volume={50},
  number={12},
  pages={3253--3259},
  year={2014},
  publisher={Elsevier}
}

@article{tegling2023scale,
  title={Scale fragilities in localized consensus dynamics},
  author={Tegling, Emma and Bamieh, Bassam and Sandberg, Henrik},
  journal={Automatica},
  volume={153},
  pages={111046},
  year={2023},
  publisher={Elsevier}
}

\end{document}